\documentclass{amsart}
\usepackage[foot]{amsaddr}
\usepackage[utf8]{inputenc}
\usepackage{latexsym,array,delarray,amsthm,amssymb,epsfig}
\usepackage{caption}
\usepackage{subcaption}
\usepackage{amscd}
\usepackage{verbatim}
\usepackage{tikz}
\usetikzlibrary{ decorations.markings}

\theoremstyle{plain}
\newtheorem{theorem}{Theorem}[section]
\newtheorem{lemma}[theorem]{Lemma}
\newtheorem{proposition}[theorem]{Proposition}
\newtheorem{corollary}[theorem]{Corollary}
\newtheorem{conjecture}[theorem]{Conjecture}

\theoremstyle{definition}
\newtheorem{definition}[theorem]{Definition}
\newtheorem{example}[theorem]{Example}

\theoremstyle{remark}
\newtheorem{remark}[theorem]{Remark}

\newcommand{\B}[1]{\mathbb #1}
\newcommand{\C}[1]{\mathcal #1}
\newcommand{\F}[1]{\mathfrak #1}

\newcommand{\BF}[1]{\mathbf #1}

\newcommand{\ind}{\mbox{$\perp \kern-5.5pt \perp$}}

\DeclareMathOperator{\im}{Im}
\DeclareMathOperator{\rank}{rank}

\DeclareMathOperator{\Aut}{Aut}

\DeclareMathOperator{\Trop}{Trop}

\title{Dimensions of Level-1 Group-Based Phylogenetic Networks}
\author{Elizabeth Gross$^1$}
\address{$^1$Department of Mathematics, University of Hawai`i at M\={a}noa, Hawai`i, USA}
\email{egross@hawaii.edu}

\author{Robert Krone$^2$}
\address{$^2$Department of Mathematics, UC Davis, California, USA}

\author{Samuel Martin$^3$}
\address{$^3$Earlham Institute, Norwich Research Park, Norwich, UK}
\email{samuel.martin@earlham.ac.uk}
\date{\today}

\begin{document}

\begin{abstract}
Phylogenetic networks represent evolutionary histories of sets of taxa where horizontal evolution or hybridization has occurred. Placing a Markov model of evolution on a phylogenetic network gives a model that is particularly amenable to algebraic study by representing it as an algebraic variety. In this paper, we give a formula for the dimension of the variety corresponding to a triangle-free level-1 phylogenetic network under a group-based evolutionary model. On our way to this, we give a dimension formula for codimension zero toric fiber products. We conclude by illustrating applications to identifiability.

\end{abstract}

\maketitle

\section{Introduction}
In evolutionary biology, phylogenetic networks are graphs used to represent the evolutionary history of a set of taxa or species. These graphs are usually paired with a statistical model where the graph is a combinatorial parameter of the model.  In this work, we focus on \emph{network-based Markov models}.  In particular, fixing a directed graph $\mathcal N$ with $n$ leaves, i.e. a network, the associated network-based Markov model is the image of a polynomial parameterization in the space of probability distributions over the sample space, which commonly in applications is $\{A, G, C, T\}^n$ where $A, G, C, T$ are the four-nucleic bases.  

We are interested in the geometry of network-based Markov models, in particular, their dimensions.  Such work is along the lines of \cite{sturmfels2005toric},\cite{eriksson2005phylogenetic},\cite{allman2007phylogenetic}, \cite{allman2008phylogenetic},\cite{casanellas2008geometry}, \cite{zwiernik2009geometry}, \cite{casanellas2011relevant}, \cite{michalek2011geometry}, \cite{casanellas2017local}, \cite{michalek2019phylogenetic}, and \cite{casanellas2021distance},  which study the geometry of tree-based Markov models. Indeed, by moving to $\mathbb C$ and taking Zariski closures, images of the parameterization maps correspond to algebraic varieties whose study can aid in model selection (see \cite{pachter2005algebraic}, \cite{drton2008lectures}, and \cite{sullivant2018algebraic} for discussions).  
Popular constraints on the parameter space, such as Jukes-Cantor (JC), Kimura 2-parameter (K2P), and Kimura 3-parameter (K3P) constraints, give rise to a class of models referred to as \emph{group-based models}. Assuming group-based constraints, the varieties associated to tree-based Markov models are toric varieties after a transformation of coordinates \cite{sturmfels2005toric}. The dimensions of tree varieties can be understood using tools from toric geometry.  While under this same transformation, group-based network varieties have a lower dimensional toric action on them, and thus are $T$-varieties (see \cite[Remark 4.1]{CHM}), these varieties are generally less well understood.  \emph{In this paper, we expand our understanding of these varieties by giving a formula for the dimension for all level-1 triangle-free group-based network varieties.}


 As described in Section \ref{sec:prelim}, a group-based model of evolution is defined with a finite abelian group $G$ and a subgroup $B$ of the automorphism group of $G$, denoted $\Aut(G)$. In a network-based Markov model, each edge of the network has a transition matrix associated to it, representing the probabilities of each type of nucleotide (usually $A$, $C$, $G$ or $T$) mutating to another over an evolutionary time interval.  The parameters of the model are the entries of these transition matrices along with a mixing parameter for each cycle.  In a group-based model, the dimension of the parameter space is cut significantly by placing constraints on the transition matrices.  In particular, each nucleotide is identified with an element of $G$, and the transition probability of a mutation from $a$ to $b$ depends only on $b - a$, reducing the number of free parameters in each matrix to $|G| - 1$.  The parameter space is reduced further by identifying the parameters for all elements of $G$ that are in the same $B$-orbit.  If $l+1$ is the number of $B$-orbits in $G,$ the number of free parameters for each edge is then $l$.
 
 For a phylogenetic network $\C N$ with $m$ edges and $c$ cycles, the {\em expected dimension} of the group-based  network variety $\dim V_{\C N}^M$ is $l(m-c)+1$, and Proposition \ref{prop:upper} shows that it is indeed an upper bound.  The main theorem of this paper shows that most level-1 group-based network varieties have the expected dimension.

 \begin{theorem} \label{thm:main} 
 Let $\C N$ be a level-1 triangle-free phylogenetic network with $n$ leaves, $m$ edges, and $c$ cycles. Let $G$ be a finite abelian group of order at least $3$ and $B$ a subgroup of $\Aut(G)$.  Let $l+1$ be the number of $B$-orbits in $G$. Then the group-based network variety $V^{(G,B)}_{\C N}$ has dimension $l(m-c)+1$.
 \end{theorem}
 
 When $G=\mathbb{Z}/2\mathbb{Z}$, the $3$ and $4$-sunlet networks do not have the expected dimension. In this case, since $\Aut(G)$ is the trivial group, there is only a single group-based model. This is the Cavender-Farris-Neyman (CFN) model, and has biological relevance, so we give the result for this group separately. Note that here we are able to give a full result for level-1 phylogenetic networks.
 
 \begin{theorem}\label{thm:zmod2z}
 Let $G=\mathbb{Z}/2\mathbb{Z}$ and let $\C N$ be a level-1 phylogenetic network with $n$ leaves, $m$ edges, $c_{\geq 5}$ cycles of length at least $5$, $c_4$ 4-cycles, and $c_3$ 3-cycles. Then the group-based network variety $V^{G}_{\C N}$ has dimension $m- (c_{\geq 5} + 2c_4 + 3 c_3) +1$.
 \end{theorem}
 
 Our main tool for proving these theorems is the toric fiber product. This is an operation on ideals that was first introduced in \cite{TFP} and generalises the Segre product. One of the first applications was to phylogenetic trees under group-based models, where the ideals of the model are toric fiber products, and the operation corresponds to the graph operation of cutting a tree at an internal edge. To some extent this remains true for phylogenetic networks and allows us to focus our attention on a family of phylogenetic networks called sunlet networks. In Section \ref{sec:TPP} we give a general dimension formula for toric fiber products (Theorem \ref{thm:tfpdim}) and apply this to phylogenetic trees and networks.

%
%
 \section{Preliminaries}\label{sec:prelim}
 
In this section, we lay out the background needed for the paper.  In particular, we review group-based models of sequence evolution where the combinatorial parameters are phylogenetic networks, as well as two tools that underlie the proof of our main theorems: tropical geometry for dimension analysis and toric fiber products. The main objects of biological relevance in this paper are phylogenetic networks, and, thus, that is where we begin.

\subsection{Phylogenetic networks}

The following network notation and terminology is adapted from \cite{francis2018new} \cite{francis2015phylogenetic}, and \cite{semple2016phylogenetic}.
 \begin{definition}
 A  \emph{(binary rooted) phylogenetic network} $\C N$ on a set $X$ is a rooted, acyclic, directed graph with no parallel edges that satisfies:
 \begin{itemize}
 \item The root vertex has outdegree 2.
 \item All vertices of outdegree 0 have indegree 1. These vertices are called \emph{leaves} and are labelled by $X$.
 \item All other vertices have either indegree 1 and outdegree 2 (called \emph{tree vertices}), or indegree 2 and outdegree 1 (called \emph{reticulation vertices}).  The incoming edges of a reticulation vertex are called \emph{reticulation edges}.
 \end{itemize}
 \end{definition}
 
A \emph{level-1} phylogenetic network is a phylogenetic network where each cycle in the underlying undirected graph contains exactly one reticulation vertex. A \emph{semi-directed network} is a mixed graph obtained from a phylogenetic network by suppressing the root node and undirecting all tree edges while the reticulation edges remain directed.  In a semi-directed network, the reticulation vertices are the vertices of indegree two and level-1 is defined the same as for a rooted phylogenetic network. A triangle-free level-1 semi-directed network is a level-1 semi-directed network where every cycle in the unrooted skeleton has length greater than three. 
For our work, it will be helpful to reduce the number of edges in a semi-directed network that we consider. To this end we introduce \emph{contracted semi-directed networks}. A contracted semi-directed network is a mixed graph obtained from a semi-directed network by contracting the non-reticulation edge of each reticulation vertex (see for example, Figure \ref{fig:example}). Note that since level-1 networks are tree-child networks, in a contracted level-1 semi-directed network, two distinct reticulation vertices are never identified, and thus each non leaf-adjacent reticulation vertex has indegree 2 and outdegree 2, and each leaf-adjacent reticulation vertex has indegree 2 and outdegree 0.   Furthermore, the level-1 condition in a contracted level-1 semi-directed network means that at least one of the outgoing edges of a reticulation vertex is a non-reticulation edge.

Finally, a $n$-\emph{sunlet network} is the semi-directed network topology with $n$ leaves and a single cycle of length $n$, where each vertex in the cycle is adjacent to a leaf vertex and one vertex in the cycle is a reticulation vertex. The $4$-sunlet network is depicted in Figure \ref{fig:example}.  Since an arbitrary level-1 network can be decomposed into a collection of trees and sunlet networks, sunlet networks will play a key role in our study.
 
 
\subsection{Group-Based Models of Evolution}
Fix an abelian group $G$ and a subgroup $B \subset \text{Aut}(G)$. Denote by $B\cdot G$ the set of $B$-orbits in $G$ and let $|B\cdot G| = l+1$. For a phylogenetic tree or network $\C N$, such a choice of $G$ and $B$ defines a model of evolution on $\C N$. From this model one can derive an algebraic variety, which we will denote $V_{\C N}^{(G,B)}$. These varieties are our primary objects of study.

 First, let us set up the notation and preliminaries for phylogenetic trees, i.e. phylogenetic networks with no reticulation vertices. For more details on group-based models on trees, see \cite[Section 15.3]{sullivant2018algebraic} and \cite{sturmfels2005toric}.  Let $\C T$ be an $n$-leaf phylogenetic tree, with vertex set, edge set, and leaf set denoted by $\C V(\C T)$, $\C E(\C T)$, and $\C L(\C T)$ respectively.  Let $m = | \C E(\C T)|$ be the number of edges in $\C T$. A \emph{consistent leaf G-labelling} of $\C T$ is a function $\xi: \C L(\C T) \longrightarrow G$ that satisfies
 $$\sum_{v\in\C L(\C T)} \xi(v) = 0.$$
 Note that the set of consistent leaf $G$-labellings depends only on $n$, and not on the edges of $\mathcal T$, so all $n$-leaf phylogenetic trees share the same set of consistent leaf $G$-labellings, which has size $|G|^{n-1}$. When $G$ is clear, we will call $\xi$ a consistent leaf labelling.
 
For a phylogenetic tree $\mathcal T$, each edge $e \in \C E(\C T)$ is oriented away from the root vertex.  Let $\C L(e) \subset \C L(\C T)$ be the set of leaves on the arrow side of $e$. A consistent leaf labelling $\xi$ of $\C T$ induces a consistent edge labelling of $\C T$ (also denoted $\xi$), which is a map $\xi:\C E(\C T) \longrightarrow G$ given by 
 
 $$\xi(e) = \sum_{v \in \C L(e)} \xi(v).$$
 
To each edge $e$ in a phylogenetic tree or network we associate $l+1$ parameters, denoted $a_e^g$, where $g$ is a representative of the $B$-orbit $[g]$. For a tree $\C T$ with $n$ leaves, the parameterization in Fourier coordinates (see \cite{sturmfels2005toric}) of the group-based model on $\C T$ is

\begin{equation}
q_{g_1g_2\cdots g_n} = \prod_{e\in \C E(\C T)} a_e^{\xi(e)},
\end{equation}
where $\xi$ is given by the consistent leaf labelling $g_1, \ldots, g_n$. Index the standard basis of $\mathbb{C}^{m(l+1)}$ with upper indices $g$ for some representatives of the orbits in $B\cdot G$, and lower indices by the edges $e \in \C E(\C T)$, and index the standard basis of $\mathbb{C}^{|G|^{n-1}}$ by the consistent leaf-labellings. The parameterization map is the map

$$ \phi_{\C T}: \mathbb{C}^{m(l+1)} \to \mathbb{C}^{|G|^{n-1}}$$
where
$$(\phi_{\C T}(w))_{g_1\cdots g_n} =  \prod_{e\in \C E(\C T)} w_e^{\xi(e)},$$
for $w \in\mathbb{C}^{m(l+1)}$ and consistent leaf labellings $\xi$ with leaf labels $g_1,\ldots, g_n$. The Zariski closure of the image of this map is called the \emph{phylogenetic variety of $\C T$ and $(G, B)$} and is denoted by $V^{(G,B)}_{\C T}$. 

Now, denote by $R$ the $\mathbb{C}$-algebra $\mathbb{C}[ q_{g_1 \cdots g_n}\ |\ g_1 +\cdots + g_n = 0]$ and by $S_{\C T}$ the $\mathbb{C}$-algebra $\mathbb{C}[a^g_e\ |\ [g] \in B\cdot G,\, e \in \C E(\C T)]$. The parameterization map $\phi_{\C T}$ is a morphism of affine varieties, with comorphism given by the $\mathbb{C}$-algebra homomorphism $ \psi_{\C T}: R \to S_{\C T}$ which acts on generators as
 \begin{equation*}
\psi_{\C T}( q_{g_1\cdots g_n}) = \prod_{e\in \C E(\C T)} a_e^{\xi(e)}.
 \end{equation*}
 It follows that the vanishing ideal of $V^{(G,B)}_{\C T}$, denoted $I_{\C T}^{(G,B)}$, is the kernel of $\psi_{\C T}$.
 
 We now move from trees to networks. Let ${\C N}$ be a level-1 phylogenetic network with $n$ leaves, $m$ edges, and $k$ reticulation vertices. Since $\C N$ is a phylogenetic network, if we remove one of the two reticulation edges for each reticulation vertex, we obtain a phylogenetic tree. We encode a choice of reticulation edge for each reticulation vertex with a vector $\sigma \in \{0,1\}^k$, and denote the resulting $n$-leaf phylogenetic tree by $\C T_\sigma$. Then the parameterization of our group based model on $\C N$ is the map $ \phi_{\C N}: \mathbb{C}^{m(l+1)} \to \mathbb{C}^{|G|^{n-1}}$ given by
  
 \begin{equation}
\phi_{\C N} = \sum_{\sigma\in\{0,1\}^k}\phi_{\C T_\sigma}
\end{equation} 
 As above, we call the Zariski closure of the image of this map the \emph{phylogenetic variety of $\C N$ and $(G, B)$}, and denote it $V^{(G,B)}_{\C N}$. The vanishing ideal $I_{\C N}^{(G,B)}$ of $V^{(G,B)}_{\C N}$ is the kernel of the $\mathbb{C}$-algebra homomorphism $\psi_{\C N}$ given by
 
  \begin{equation}\label{eqn:param}
 \begin{aligned} 
 \psi_{\C N}: R &\to S_{\C N}\\
 q_{g_1\cdots g_n}&\mapsto \sum_{\sigma\in\{0,1\}^k}\psi_{\C T_\sigma}( q_{g_1\cdots g_n}).
  \end{aligned}
 \end{equation}
where $S_{\C N} = \mathbb{C}[a_e^g\ |\ [g] \in B\cdot G,\, e \in \C E(\C N)]$ and we identify $S_{\C T_\sigma}$ as a subalgebra of $S_{\C N}$ in the obvious way.
 
When $B = \{\rm{id}\}$, we call the probabilistic model associated to $(G, B)$ the \emph{general group-based model} for the group $G$ and denote the corresponding variety as $V_{\C N}^G =V_{\C N}^{(G, B)}$. The K3P model is the general group-based model for the Klein-4 group, and the CFN model is the general group-based model for the group $\mathbb{Z}/2\mathbb{Z}$.  The pairs $(G, B)$ corresponding to JC and K2P are $(\B Z/2\B Z \times \B Z/2\B Z , \F S_3)$ and $(\B Z/2\B Z \times \B Z/2\B Z , \F S_2)$, respectively, where $\Aut(\B Z/2\B Z \times \B Z/2\B Z )$ is identified with the permutation group 
$\F S_3$.

In this paper, we are concerned with $\dim V_{\C N}^{(G, B)}$ for a general $\C N$ and $(G,B)$.  In previous work, it is shown that under the CFN, JC, K2P, and K3P models, two phylogenetic network varieties are the same if the two networks have the same underlying semi-directed network \cite{GL} \cite{gross2021distinguishing}. Here, we extend these results to all group-based models.  This allows us to focus our attention on semi-directed networks. 
 
 \begin{lemma}\label{lemma:semidirectedmodel}
 Let $G$ be a finite abelian group and let $B$ be a subgroup of $\Aut(G)$. If $\C N_1$ and $\C N_2$ are two phylogenetic networks with the same underlying semi-directed network , then $V_{\C N_1}^{(G,B)} = V_{\C N_2}^{(G,B)}$, where equality here means equality as sets.
 \end{lemma}
 \begin{proof}
Since phylogenetic networks have no parallel edges, it is clear that two phylogenetic networks $\C N_1$ and $\C N_2$ have the same semi-directed network if and only if their corresponding unrooted networks differ only by the directions of their non-reticulation edges. Therefore it is sufficient to show that suppressing vertices of degree 2 and changing the orientation of a single non-reticulation edge do not affect the model.

First notice that if we take a phylogenetic network $\mathcal N_1$ and reorient any collection of the non-reticulation edges to form a new network $\mathcal N_2$ (not necessarily a phylogenetic network), the maps described above are still well-defined and so too are the corresponding varieties. Thus, for this proof, we will relax the definition of phylogenetic networks to include such networks, as it will allow us to consider redirecting one edge at a time. Let $\C N_1$ and $\C N_2$ be two phylogenetic networks that are equal except for the direction of a single non-reticulation edge $e$, and let $\xi_1$ and $\xi_2$ be consistent leaf labellings on $\C N_1$ and $\C N_2$ respectively, with the property that $\xi_1(v) = \xi_2(v)$ for each leaf $v$ in the skeleton of $\C N_1$ and $\C N_2$. Then it is clear that $\xi_1(e) = -\xi_2(e)$. This means that $\phi_{\C N_2} = \phi_{\C N_1} \circ \theta$, where $\theta$ is the automorphism of $\mathbb{C}^{m(l+1)}$ given by swapping the coefficients corresponding to $a_e^{g}$ and $a_e^{-g}$ whenever $g \neq -g$. Since $\theta$ is bijective, composing $\theta$ with $\phi_{\C N_1}$ does not affect the image of $\phi_{\C N_1}$, so it follows that $V_{\C N_1}^{(G,B)} = V_{\C N_2}^{(G,B)}$.

Next let $\C N$ be a phylogenetic network with a vertex $v$ of order 2 that has incident edges $e_1$ and $e_2$. Let $\xi$ be a consistent leaf labelling of $\C N$ with  $\xi(e_1) = g$ so that either $\xi(e_2) = g$ or $\xi(e_2) = -g$. Let us suppose that $\xi(e_2) = g$ and note that the proof in the other case is similar. Let $\C N'$ be the phylogenetic network got from $\C N$ by suppressing $v$. Denote the new edge of $\C N'$ by $e'$ and, without loss of generality, give $e'$ the same orientation as $e_1$ so that $\xi(e') = g$. Let $\theta: \mathbb{C}^{(m-1)(l+1)} \to \mathbb{C}^{m(l+1)}$ be the map from the parameter space of $\C N'$ to the parameter space of $\C N$ that is constant on all parameters for edges shared by $\C N$ and $\C N'$, takes all parameters for edge $e'$ to the corresponding parameter for edge $e_1$, and sets all parameters corresponding to edge $e_2$ to $1$. Then it is clear that we have
$$ \phi_{\C N'} = \phi_{\C N} \circ\theta,$$
and thus $\im{\phi_{\C N'}} \subseteq \im{\phi_{\C N}}$. On the other hand, consider $\phi_{\C N} (w)$ for $w \in \mathbb{C}^{m(l+1)}$. Taking $u\in\mathbb{C}^{(m-1)(l+1)}$ such that $u_{e'}^g = w_{e_1}^{g} w_{e_2}^{g}$ for all $[g]\in B\cdot G$ and $u_d^g = w_d^g$ for all edges $d \neq e'$ and all $[g]\in B\cdot G$, we see that $\phi_{\C N'} (u) = \phi_{\C N}(w)$. It follows that $\im{\phi_{\C N'}} = \im{\phi_{\C N}}$ and thus $V_{\C N'}^{(G,B)} = V_{\C N}^{(G,B)}$.

 \end{proof}
 
 
 Note that since the orientation of the non-reticulation edges does not affect the variety, we may choose any orientation for non-reticulation edges, even if this is not consistent with any placement of a root vertex. Thus when considering a phylogenetic network $\C N$, we may take the corresponding semi-directed network and arbitrarily assign orientations to each non-reticulation edge to obtain a parameterization of the model.

 \begin{example}
 Let $G$ be a finite abelian group and $B = \{\rm{id}\}$. Let $\C N$ be a $4$-sunlet network with leaf labellings, edge orientations, and edge labellings as in Figure \ref{fig:example}. The map $\psi_{\C N}$ is given by
 
 $$ q_{g_1 g_2 g_3 g_4} \longmapsto a_1^{g_1}a_2^{g_2}a_3^{g_3}a_4^{g_4}a_5^{g_1}a_6^{g_1 + g_2}a_7^{g_4} + a_1^{g_1}a_2^{g_2}a_3^{g_3}a_4^{g_4}a_6^{g_2}a_7^{g_1 + g_4}a_8^{g_1} ,$$
where to simplify notation we write $a_i^g$ for $a_{e_i}^g$. Here, the first monomial corresponds to the tree obtained by removing the edge $e_8$, and the second monomial corresponds to the tree obtained by removing the edge $e_5$.
 \end{example}
 
 \begin{figure}[h!]
 \centering
 \resizebox{1.0\textwidth}{!}{
\begin{tikzpicture}
[every node/.style={inner sep=0pt},
                    every path/.style={thick},   
decoration={markings, 
    mark= at position 0.5 with {\arrowreversed{stealth}}
    }
]
                    
\draw [postaction={decorate}] (0,2) -- (2,2) node[midway,below=4] {$e_{4}$};
\draw [postaction={decorate}] (2,2)  -- (4 ,4)  node[midway,above=8,left=0.5] {$e_{7}$};
\draw [postaction={decorate}] (8,2) -- (6,2) node[midway,below=4] {$e_{2}$};
\draw [postaction={decorate}] (6,2)  -- (4 , 4)  node[midway,above=8,right=0.5] {$e_{6}$};
\draw [postaction={decorate}] (4,5.5)  -- (4 , 4)  node[midway,above=8,right=1] {$e_{3}$};
\draw [postaction={decorate}] (4,-1.5) -- (4,0) node[midway,below=0,right=2] {$e_1$};
\draw[dashed,postaction={decorate}] (4,0)  -- (2,2) node[midway,below=4,left=2] {$e_{8}$};
\draw[dashed,postaction={decorate}] (4,0)  -- (6,2) node[midway,below=4,right=2] {$e_5$};

\node (v3) at (4,-1.85) {$1$};
\node (v3) at (8.35,2) {$2$};
\node (v3) at (4,5.85) {$3$};
\node (v3) at (-0.35,2) {$4$};

\draw [postaction={decorate}] (10,2) -- (12,2) node[midway,below=4] {$e_{4}$};
\draw [postaction={decorate}] (12,2)  -- (14 ,4)  node[midway,above=8,left=0.5] {$e_{7}$};
\draw [postaction={decorate}] (18,2) -- (16,2) node[midway,below=4] {$e_{2}$};
\draw [postaction={decorate}] (16,2)  -- (14, 4)  node[midway,above=8,right=0.5] {$e_{6}$};
\draw [postaction={decorate}] (14,5.5)  -- (14 , 4)  node[midway,above=8,right=1] {$e_{3}$};
\draw[dashed,postaction={decorate}] (14,0)  -- (12,2) node[midway,below=4,left=2] {$e_{8}$};
\draw[dashed,postaction={decorate}] (14,0)  -- (16,2) node[midway,below=4,right=2] {$e_5$};

\node (v3) at (14,-0.35) {$1$};
\node (v3) at (18.35,2) {$2$};
\node (v3) at (14,5.85) {$3$};
\node (v3) at (9.65,2) {$4$};

\end{tikzpicture}
}
\caption{A leaf-labelled, directed $4$-sunlet network (left), and its corresponding contracted network (right).}
\label{fig:example}
\end{figure}
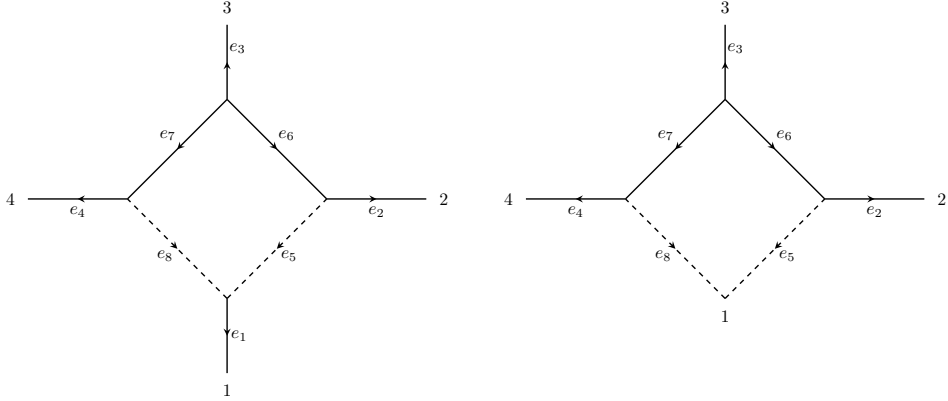

To end this section, we show that sunlet networks and contracted sunlet networks have the same corresponding varieties.
 
 \begin{lemma}\label{lemma:contractedNetwork}
Let $G$ be a finite abelian group and $B$ a subgroup of $\text{Aut}(G)$.  Let $\C N$ be a sunlet network and let $\C N'$ be its contraction. Then $V_{\C N}^{(G,B)} = V_{\C N'}^{(G,B)}$. 
 \end{lemma}
 \begin{proof}
Let $\C N$ be an $n$ -sunlet, so that $\C N$ has $2n$ edges, denoted $e_1$,\ldots,$e_{2n}$, with $e_1$ the leaf edge adjacent to the reticulation vertex, and $e_{n+1}, e_{2n}$ the two reticulation edges, as in Figure \ref{fig:example}. Let $\phi_{\C N}$ denote the parameterization map for $\C N$, and let $\phi_{\C N'}$ denote the parameterization map for its contraction $\C N'$. As before, index the parameter spaces $\mathbb{C}^{2n(l+1)}$ and $\mathbb{C}^{(2n-1)(l+1)}$ by the $B$-orbits in $G$ and the edges of $\C N$ and $\C N'$ respectively.
 
 It is clear that $\phi_{\C N'} = \phi_{\C N} \circ \iota$, where $\iota:\mathbb{C}^{(2n-1)(l+1)} \rightarrow\mathbb{C}^{2n(l+1)}$ is given by $\iota(w)_e^g = w_e^g$ for all $g\in B\cdot G$ and $e \in \C E(\C N')$ (i.e. $e \neq e_1$), and $\iota(w)_{e_1}^g = 1$ for all $g \in B\cdot G$.  It follows that $\im \phi_{\C N'} \subseteq \im \phi_{\C N}$.
 
 On the other hand, let $y=\phi_{\C N}(w)$. Let $u\in\mathbb{C}^{(2n-1)(l+1)} $ be given by $u_e^g = w_e^g$ for $e\neq e_{n+1}, e_{2n}$, and $u_{e_{n+1}}^g = w_{e_1}^g w_{e_{n+1}}^g$, and $u_{e_{2n}}^g = w_{e_1}^g w_{e_{2n}}^g$ for all $g \in B\cdot G$. It follows that $\phi_{\C N'}(u) = y = \phi_{\C N}(w)$, so $\im \phi_{\C N} \subseteq \im \phi_{\C N'}$.
 \end{proof}
 
 In fact, by absorbing the parameters associated to the contracted edge of each reticulation vertex into the corresponding parameters of both reticulation edges, the above lemma can be extended for any level-1 phylogenetic network.
 
  \begin{lemma}\label{lemma:contractedLevelOneNetwork}
Let $G$ be a finite abelian group and $B$ a subgroup of $\text{Aut}(G)$.  Let $\C N$ be a level-1 semi-directed and let $\C N'$ be its contraction. Then $V_{\C N}^{(G,B)} = V_{\C N'}^{(G,B)}$. \qed
 \end{lemma}
 
\subsection{Tropical Geometry}\label{subsec:tropical}

In Section \ref{sec:sunlet} we give a lower bound on the dimension of the sunlet varieties by using the tropical geometry results of \cite{Draisma}. Here we will present the result tailored to our needs.

Let $C$ be a Zariski-closed cone in a complex vector space $V$ of dimension $n$, and let $W$ be a complex vector space of dimension $m$ such that we have a  polynomial map $f:W\to V$ mapping $W$ dominantly to $C$.  A classic result is that the rank of the Jacobian matrix of $f$ at any point in $x \in W$ gives a lower bound on $\dim(C)$ (and equality holds when $x$ is generic).  We similarly obtain a bound from the tropicalization of $f$.

Fix bases of $V$ and $W$ so that we may write $f = (f_b)_{b= 1}^{n}$, where $f_b \in \mathbb{C}[x_1, \ldots, x_m]$ for $b = 1,\ldots n$. Write 
 $$ f_b = \sum_{\alpha \in M_b} c_\alpha x^\alpha$$
for the finite subset $M_b\subset\mathbb{Z}_{\ge 0}^m$ consisting of those $\alpha$ for which $c_\alpha \neq 0$.  The tropicalization of $f_b$ is defined as the piece-wise linear function $\Trop(f_b): \mathbb{R}^m\to\mathbb{R}$ given by
 
 $$ \Trop(f_b)(\lambda) := \min_{\alpha \in M_b} \langle \lambda, \alpha \rangle,$$
for $\lambda\in \mathbb{R}^m$.  Then $\Trop(f): \mathbb{R}^m \to \mathbb{R}^n$ is defined as $(\Trop(f_b))_{b = 1}^n$.

We will not define here the tropical variety $\Trop(C) \subseteq \mathbb{R}^n$, but we note two relevant facts (see e.g. \cite{maclagan2021introduction}):
\begin{itemize}
    \item $\Trop(C)$ is a polyhedral complex with dimension bounded by $\dim(C)$,
    \item and $\im(\Trop(f)) \subseteq \Trop(C)$.
\end{itemize} Therefore the Jacobian of $\Trop(f)$ at a point $\lambda \in \mathbb{R}^m$ where the map is differentiable gives a lower bound on $\dim(C)$ (although it is no longer true that equality necessarily holds when $\lambda$ is generic).

Fix $\lambda$ such that $\Trop(f)$ is differentiable at $\lambda$, meaning that $\Trop(f)$ is linear in an open neighborhood $U$ of $\lambda$.  Specifically $\Trop(f_b)(\mu) = \langle \mu, \alpha^\prime_b\rangle$ for all $\mu \in U$ where $\alpha^\prime_b$ is the unique vector in $M_b$ that minimizes $\langle \lambda, \alpha^\prime_b\rangle$.  Then $\Trop(f)(\mu) = A_\lambda^T \mu$ where $A_\lambda$ is the $m \times n$ matrix with columns $\alpha^\prime_1, \ldots, \alpha^\prime_n$.  (Note that $A_\lambda^T$ is also the Jacobian matrix of $\Trop(f)$ at $\lambda$.)  The lemma below follows.

 \begin{lemma}\cite[Corollary~2.3]{Draisma}\label{lemma:draisma}
 Let the notation be as above. Then
 
 $$\dim C \geq \max_{\lambda\in\mathbb{R}^m} \rank_{\mathbb{R}}A_\lambda.$$ 
 \end{lemma}
 
 For our purposes, $f$ will be given by the polynomial parameterization map $\phi_{\C N}$. Since the variety $V^{(G,B)}_{\C N}$ is equal to the Zariski closure of $\phi_{\C N}(\mathbb{C}^{|G|^{n-1}})$, and since each polynomial in the parameterization is homogeneous, $V^{(G,B)}_{\C N}$ is a closed cone.

 \subsection{Toric Fiber Products}
 The toric fiber product is an algebraic operation that takes two homogeneous ideals with compatible multigradings and produces a new homogeneous ideal. It was introduced in \cite{TFP} in order to generalise the gluing operation for toric ideals that appear in tree-based models of evolution and elsewhere in algebraic statistics, and further studied in \cite{engstrom2014multigraded} and \cite{kahle2014toric}. More recently toric fiber products were introduced into the geometric modelling setting in \cite{duarte2023toric}. Here, we will introduce the basic objects and recommend that the reader consult \cite{TFP} for further details. 
 
 Let $r\in\mathbb{N}$ and $s,t\in\mathbb{N}^r$, and let $\C A = \{a_1, \ldots, a_r\} \subset \mathbb{Z}^D$ be a linearly independent set for some $D > 0$. Denote the affine semigroup generated by $\C A$ by $\mathbb{N}\C A$. Let $\mathbb{K}$ be an algebraically closed field and let
 
$$ \mathbb{K}[x] = \mathbb{K}[x_j^i\ |\ i \in [r], j \in [s_i]],$$
and
 
 $$ \mathbb{K}[y] = \mathbb{K}[y_k^i\ |\ i \in [r], k \in [t_i]],$$
 be multigraded polynomial rings with multidegree given by $\text{deg}(x_j^i) = \text{deg}(y_k^i) = a_i$ for all $i=1,\ldots, n$, $j = 1,\ldots, s_i$, and $k=1,\ldots,t_i$. Note that since the $a_i$ are linearly independent, ideals in $\mathbb{K}[x]$ and $\mathbb{K}[y]$ that are homogeneous with respect to the multigrading are also homogeneous with respect to the total degree. For homogeneous ideals $I\subset\mathbb{K}[x]$ and $J\subset\mathbb{K}[y]$, let $R= \mathbb{K}[x]/I$ and $S = \mathbb{K}[y]/J$ be the corresponding quotient rings, which inherit the multigrading from $\mathbb{K}[x]$ and $\mathbb{K}[y]$ respectively. Let 
 
 $$\mathbb{K}[z] = \mathbb{K}[z^i_{jk}\ |\ i\in[r], j\in[s_i], k\in[t_i]]$$
 be the polynomial ring with the analogous multigrading (i.e. $\text{deg}\,z^i_{jk} = a_i$). Let $\phi_{IJ}$ be the ring homomorphism given by
 \begin{equation*}
  \begin{aligned} 
 \phi_{IJ}:  \mathbb{K}[z]&\longrightarrow R\otimes_{\mathbb{K}} S \\
 		 z_{jk}^i &\longmapsto x_j^i \otimes y_k^i.
\end{aligned}
 \end{equation*}
 
 \begin{definition}
With notation as above, the \emph{toric fiber product} of $I$ and $J$ is
 
 $$I\times_{\C A} J := \text{ker}\,\phi_{IJ}.$$
 \end{definition}
 Note that when $I$ and $J$ are prime ideals, since the toric fiber product $I\times_{\C A} J$ is the kernel of a ring homomorphism into an integral domain, it is a prime ideal. It will be helpful for us to also consider the monomial homomorphism
  \begin{equation*}
  \begin{aligned} 
 \phi_{B}:  \mathbb{K}[z]&\longrightarrow \mathbb{K}[x,y] \\
 		z_{jk}^i &\longmapsto x_j^i y_k^i,
\end{aligned}
 \end{equation*}
where we think of $B$ as being the integral matrix of exponent vectors of this map. The ideal $I_B = \text{ker}\ \phi_B$ is a toric ideal and is given by

$$I_B = \big\langle\text{Quad}_B\big\rangle,$$
where 
 $$\text{Quad}_B = \{z^i_{j_1k_2}z^i_{j_2k_1} - z^i_{j_1k_1}z^i_{j_2k_2}\ |\ i\in [r],\ 1\leq j_1< j_2 \leq s_i,\ 1\leq k_1< k_2 \leq t_i\},$$
 and $\text{Quad}_B$ is a Gr\"{o}bner basis for $I_B$ with respect to any term order that selects the first term (as written above) as the initial term for each quadric \cite[Proposition~10]{TFP}.

We may also define $I\times_{\C A}J$ as $\phi_{B}^{-1}(I + J)$, where we consider $I$ and $J$ as being their natural extensions in $\mathbb{K}[x,y]$. If $\omega_1$ and $\omega_2$ are weight vectors on $\mathbb{K}[x]$ and $\mathbb{K}[y]$ respectively, then we have a natural weight vector $(\omega_1,\omega_2)$ on $\mathbb{K}[x,y]$, and the pullback $\phi_B^*(\omega_1,\omega_2)$  is a weight vector on $\mathbb{K}[z]$. These weight vectors have the property that for all monomials $z^a \in \mathbb{K}[z]$ we have $\text{wt}_{\phi_B^*(\omega_1, \omega_2)}(z^a) = \text{wt}_{(\omega_1, \omega_2)}(\phi_B(z^a))$.
 
 Let $f \in\mathbb{K}[x]$ be a homogeneous polynomial with respect to the multigrading $\mathbb{N} \C A$ and total degree $d$,  so that we may write
 
 $$f = \sum_{u = 1}^v c_u x^{i_1}_{j_1^u}\cdots x^{i_d}_{j_d^u},$$
 with each $j_l^u \in [s_{i_l}]$ and $c_u \in \mathbb{K}$. The upper indices $i_1, \ldots, i_d$ can be written independent of $u$ since $f$ is homogeneous with respect to $\mathbb{N} \C A$ and $\C A$ is linearly independent. For any $k=(k_1,\ldots, k_d)$ with $k_l \in [t_{i_l}]$ define $f_k \in \mathbb{K}[z]$ by
 
 \begin{equation}
 \label{eqn:f_k}
 f_k = \sum_{u=1}^v c_u z^{i_1}_{j_1^uk_1}\cdots z^{i_d}_{j_d^uk_d}.
 \end{equation}
 
 For a set $F\subset\mathbb{K}[x]$ define $\text{Lift}\,F$ to be the subset of $\mathbb{K}[z]$ consisting of all possible $f_k$ with $f\in F$.  We define $\text{Lift}\,G$ for $G\subset\mathbb{K}[y]$ analogously. Observe that we have
 
 $$\phi_B(f_k) = \sum_{u=1}^v c_u x^{i_1}_{j_1^u}\cdots x^{i_d}_{j_d^u}y^{i_1}_{k_1}\cdots y^{i_d}_{k_d} = (y^{i_1}_{k_1}\cdots y^{i_d}_{k_d})f.$$
 Since the weight of each monomial in $f_k$ with respect to $\phi_B^*(\omega_1, \omega_2)$ is equal to the weight with respect to $(\omega_1, \omega_2)$ of the image of that monomial under $\phi_B$, and this is in turn given by the sum of the weight with respect to $\omega_1$ of the corresponding monomial in $f$ and the weight with respect to $\omega_2$ of $y^{i_1}_{k_1}\cdots y^{i_d}_{k_d}$, we have that $\text{in}_{\phi_B^*(\omega_1. \omega_2)}(f_k) = (\text{in}_{\omega_1}(f))_k$. It follows that $\text{in}_{\phi_B^*(\omega_1. \omega_2)}(\langle\text{Lift }F\rangle) = \text{Lift}(\text{in}_{\omega_1}(\langle F\rangle)$, and by symmetry $\text{in}_{\phi_B^*(\omega_1. \omega_2)}(\langle\text{Lift }G\rangle) = \text{Lift}(\text{in}_{\omega_2}(\langle G\rangle)$.
 
 One of the key results on toric fiber products is the following.
 
 \begin{theorem}\cite[Theorem~13]{TFP}\label{thm:TFPBasis}
 Let $F$ be a homogeneous Gr\"{o}bner basis for $I$ with respect to a weight vector $\omega_1$, let $G$ be a homogeneous Gr\"{o}bner basis for $J$ with respect to a weight vector $\omega_2$, and let $\omega_q$ be a weight vector such that $\text{Quad}_B$ is a Gr\"{o}bner basis for $I_B$. Then
 
 $$\text{Lift }(F)\cup\text{Lift }(G)\cup\text{Quad}_B$$
 
\noindent is a Gr\"{o}bner basis for $I \times_{\C A} J$ with respect to the weight vector $\phi_B^*(\omega_1, \omega_2) + \varepsilon\omega_q$ for sufficiently small $\epsilon > 0$.
 \end{theorem}
 Note that if $\varepsilon$ is chosen small enough, then $\text{in}_{\phi_B^*(\omega_1, \omega_2) + \varepsilon\omega_q}(f_k) = \text{in}_{\phi_B^*(\omega_1, \omega_2)}(f_k)$ for all $f_k \in \langle \text{Lift }F, \text{Lift }G\rangle$. 
 
\begin{remark}\label{rem:omega}
Since $\text{Quad}_B \subset \ker\phi_B$ we have that $\text{in}_{\phi_B^*(\omega_1,\omega_2)}(f) = f$ for all $f\in\text{Quad}_B$. Now $\text{in}_{\phi_B^*(\omega_1, \omega_2) + \varepsilon\omega_q}(f) = \text{in}_{\omega_q}(\text{in}_{\phi_B^*(\omega_1, \omega_2)}(f)) = \text{in}_{\omega_q}(f)$, so it follows that on $\text{Quad}_B$, the weight vector $\phi_B^*(\omega_1, \omega_2) + \varepsilon\omega_q$ chooses the same leading term as the weight vector $\omega_q$.
\end{remark}

%
%
 \section{Dimension of Toric Fiber Products}\label{sec:TPP}
 In this section we give a dimension formula for the toric fiber product of two prime ideals when the set $\C A$ is linearly independent, and then apply this to level-1 phylogenetic networks.
 
 Recall the following definitions, from e.g. \cite{BeckerWeispfenning}. Let $I$ be an ideal in the polynomial ring $\mathbb{K}[x_1,\dots, x_n]$. We say that a set $U\subseteq\{x_1,\ldots, x_n\}$ is \emph{independent modulo} $I$ if $I\cap\mathbb{K}[U] = \{0\}$. We say that $U$ is \emph{maximally independent modulo} $I$ if it is independent modulo $I$ and there exists no other set $U'\subseteq\{x_1,\ldots,x_n\}$ such that $U\subseteq U'$ and $U'$ is independent modulo $I$. The \emph{dimension} of $I$, denoted $\dim I$, is given by $\text{max}\{|U|\ |\ U\subseteq \{x_1,\ldots,x_n\} \text{ is independent modulo }I\}$. If $I$ is a prime ideal then for all sets $U\subseteq\{x_1,\ldots,x_n\}$ that are maximally independent modulo $I$ we have $\dim I = |U|$ . We begin with the following lemma.
 
 \begin{lemma}\label{lemma:lift}
 Let $M\subset\mathbb{K}[x]$ be a set of monomials, and let $U\subseteq\{x^i_{j_i}\ |\ i\in [r], j_i \in [s_i]\}$ be maximally independent modulo $\langle M\rangle$, given by
 $$U = \{x^i_{j_i^h}\ |\ i\in\C I,\ h=1,\ldots,n_i\},$$
 where $\C I\subset [r]$ and for each $i \in \C I$ we have $j_i^h \in [s_i]$ for $h =1,\ldots,n_i$. Then the set 
 $${\rm Lift}\,U = \{z^i_{j_i^hk}\ |\ i\in\C I,\ h=1,\ldots,n_i,\ k\in [t_i]\} \subseteq \mathbb{K}[z]$$ 
 is maximally independent modulo $\langle{\rm Lift}\,M\rangle$.
 \end{lemma}
 \begin{proof}
 First observe that $\langle\text{Lift}\,M\rangle$ is a monomial ideal generated by monomials of the form $m_k$ as in equation \eqref{eqn:f_k} for $m\in M$. Thus, in order to show independence, it is sufficient to only consider monomials $m_k$. Now if $m_k \in \langle\text{Lift}\,M\rangle\cap\mathbb{K}[z^i_{j_i^hk}\ |\ i\in\C I,\ h=1,\ldots,n_i,\ k\in [t_i]]$, then $m\in M\cap\mathbb{K}[x^i_{j_i}\ |\ i\in\C I, j_i \in [s_i]] = \{0\}$, and thus $\text{Lift}\, U$ is independent modulo $\langle{\rm Lift}\,M\rangle$. Furthermore, if $\text{Lift}\, U$ is not maximal then there exists some $i'\in[r], j'\in[s_{i'}]$, and $k'\in[t_{i'}]$ such that $\text{Lift}\, U \cup \{z^{i'}_{j'k'}\}$ is independent modulo $\langle\text{Lift}\,M\rangle$. But then  $\text{Lift}\, U \cup \{z^{i'}_{j'k}\ |\ k\in[t_{i'}]\} = \text{Lift}\, (U\cup \{x^{i'}_{j'}\})$ is also independent modulo  $\langle\text{Lift}\,M\rangle$, so $U\cup\{x^{i'}_{j'}\}$ is independent modulo $\langle M\rangle$, contradicting the maximality of $U$.
\end{proof}

Note that we have the analogous result for a set of monomials $M\subset\mathbb{K}[y]$ and  $U\subseteq\{y^i_{k_i}\ |\ i\in [r], k_i \in [t_i]\}$. 
  
 \begin{theorem}\label{thm:tfpdim}
 Let $I$ and $J$ be homogeneous ideals in $\mathbb{K}[x]$ and $\mathbb{K}[y]$ respectively, let $\omega_1$ be a weight vector for $\mathbb{K}[x]$, and let $\omega_2$ be a weight vector for $\mathbb{K}[y]$. Let the set $\{x^i_{j^h_i}\ |\ i \in \C I_1,\ j^h_i\in [s_i],\ h=1,\ldots,n_i\}$  be a maximally independent modulo $\text{in}_{\omega_1}(I)$ for some $\C I_1 \subseteq [r]$, and let the set $\{y^i_{k^g_i}\ |\ i \in \C I_2,\ k^g_i\in [t_i],\ g=1,\ldots,m_i\}$ be a maximally independent modulo  $\text{in}_{\omega_2}(J)$ with $\C I_2 \subseteq [r]$. If the set $\C A$ is linearly independent, then 
 $$\dim I \times_{\C A} J \geq \sum_{i\in \C I_1\cap\C I_2} (n_i + m_i - 1).$$
Furthermore, if $I$ and $J$ are prime and we have $\C I_1=\C I_2 = [r]$ then $\dim I \times_{\C A} J = \dim I + \dim J - |\C A|$.
 \end{theorem}
 
 \begin{proof}
Let $F$ and $G$ be Gr\"{o}bner bases of $I$ and $J$ with respect to the weight vectors $\omega_1$ and $\omega_2$ respectively, and let $\omega_q$ be a weight vector on $\mathbb{K}[z]$ that for all $i$ chooses $z^i_{j_1k_2}z^i_{j_2k_1}$ as the initial term for each polynomial in $\text{Quad}_B$, where $1\leq j_1<j_2\leq s_i$ and $1\leq k_1 < k_2\leq t_i$. By Theorem \ref{thm:TFPBasis}, we have that for the weight vector $\omega=\phi_B^*(\omega_1, \omega_2) + \varepsilon\omega_q$ and sufficiently small $\varepsilon > 0$, the set $\text{Lift}\,(F) \cup \text{Lift}\,(G) \cup \text{Quad}_B$ is a Gr\"{o}bner basis of $I\times_{\C A} J$. To prove the first statement of the theorem it is sufficient to find a set of generators $z_{j_l k_l}^{i_l}$ that are maximally independent modulo $\text{in}_\omega(I\times_{\C A} J)$ and that has size $ \sum_{i\in \C I_1\cap\C I_2} (n_i + m_i - 1)$.


As in the statement of the theorem, let the set $\{x^i_{j^h_i}\ |\ i \in \C I_1,\ j^h_i\in [s_i],\ h=1,\ldots,n_i\}$  be maximally independent modulo $\text{in}_{\omega_1}(I) = \text{in}_{\omega_1}(\langle F\rangle)$, and for each $i\in\C I_1$ arrange the $j^h_i$ so that $j^1_i < j^2_i < \cdots <j^{n_i}_i $. By Lemma \ref{lemma:lift}, and since $\text{in}_\omega(\langle\text{Lift} F\rangle) = \text{Lift}(\text{in}_{\omega_1}\langle F\rangle )$, we have that the set 
 $$\{z^{i}_{j^h_ik}\ |\ i\in\C I_1,\ h=1,\ldots,n_i,\ k=1,\ldots,t_i\} \subset \mathbb{K}[z]$$
is maximally independent modulo $\text{in}_\omega(\langle\text{Lift}\, F\rangle)$. Similarly, since the set $\{y^i_{k^g_i}\ |\ i \in \C I_2,\ k^g_i\in [t_i],\ g=1,\ldots,m_i\}$ is maximally independent modulo  $\text{in}_{\omega_2}(J)$, we have that
$$\{z^i_{jk_i^g}\ |\ i\in\C I_2,\ g=1,\ldots, m_i,\ j=1,\ldots,s_i \}\subset \mathbb{K}[z]$$
is maximally independent modulo $\text{in}_\omega(\langle\text{Lift}\, G\rangle)$. Again, for each $i\in\C I_2$ arrange the $k^h_i$ so that $k^1_i < k^2_i < \cdots <k^{m_i}_i $.  We now have 
 $$\text{in}_\omega(\langle\text{Lift}\, F \cup\text{Lift}\, G\rangle)\cap\mathbb{K}[z^i_{j_i^hk_i^g}\ |\ i\in\C I _1\cap\C I_2,\ h=1,\ldots,n_i,\ g=1,\ldots,m_i] = \{0\},$$ 
 and that the set $\{z^i_{j_i^hk_i^g}\ |\ i\in\C I_1\cap\C I_2,\ h=1,\ldots,n_i,\ g=1,\ldots,m_i\}$ is maximal with respect to this condition. We claim that the set 
 $$Z = \{z^i_{j_i^hk_i^1}\ |\ i\in\C I_1\cap\C I_2,\ h=1,\ldots, n_i\}\cup\{z^i_{j^{n_i}_ik_i^g}\ |\ i\in\C I_1\cap\C I_2,\ g=1,\ldots,m_i\}$$
 is maximally independent modulo $\text{in}_\omega(\langle\text{Lift}\,F \cup \text{Lift}\,G \cup \text{Quad}_B\rangle) = \text{in}_\omega(I\times_{\C A} J)$ (see Figure \ref{fig:tfp}).
 
 First we show that $\text{in}_\omega(\langle\text{Lift}\, F \ \cup \ \text{Lift}\, G \ \cup \ \text{Quad}_B\rangle)\cap\mathbb{K}[Z] = \{0\}$. Observe that since $\text{Lift}\, F \cup\text{Lift}\, G\cup\text{Quad}_B$ is a Gr\"{o}bner basis, we have $\text{in}_\omega(\langle\text{Lift}\, F \cup\text{Lift}\, G\cup\text{Quad}_B\rangle) = \text{in}_\omega(\text{Lift}\, F)\cup\text{in}_\omega(\text{Lift}\, G)\cup\text{in}_\omega(\text{Quad}_B)$. Since $\text{in}_\omega(\langle\text{Lift}\, F \cup\text{Lift}\, G\rangle)\cap\mathbb{K}[Z] = \{0\}$, it is sufficient to show that for each $i\in\C I_1\cap\C I_2$ the elements of degree $a_i$ in $Z$ do not appear together as a quadratic monomial in  $\text{in}_\omega(\langle\text{Quad}_B\rangle)$. By Remark \ref{rem:omega} and our choice of $\omega$ we have that 
 $$\text{in}_\omega(\text{Quad}_B) = \{ z^i_{j_1k_2}z^i_{j_2k_1}\ |\ i\in [r],\ 1\leq j_1<j_2\leq s_i, \ 1\leq k_1 < k_2\leq t_i\}.$$
 Fix $i\in \C I_1\cap\C I_2$ and observe that for any two elements of $Z$ of degree $a_i$, say $z^i_{jk}$ and $z^i_{j'k'}$ with $j \leq j'$, we either have $j = j'$, $k=k'$, or $k < k'$. In all cases $z^i_{jk}z^i_{j'k'}\not\in\text{in}_\omega(\text{Quad}_B)$.
 
 Next we show that $Z$ is maximal. By the maximal independence of $\mathbb{K}[z^i_{j_i^hk_i^g}\ |\ i\in\C I _1\cap\C I_2,\ h=1,\ldots,n_i,\ g=1,\ldots,m_i]$ modulo $\text{in}_\omega(\langle\text{Lift}\, F \cup\text{Lift}\, G\rangle)$, we need only consider those $z^i_{jk}\not\in Z$ with $i\in\C I_1\cap\C I_2$, $j=j_i^h$ for some $h=1,\ldots,n_i$, and $k=k_i^g$ for some $g=1,\ldots,m_i$. But it is clear that for any such $z^i_{jk}$, we can find $z^i_{j_0k_0} \in Z$ such that $z^i_{j_0k_0}z^i_{j_h k_l}\in\text{in}_\omega(\text{Quad}_B)$. It follows that $Z$ is maximally independent modulo $\text{in}_\omega(\langle\text{Lift}\, F \cup\text{Lift}\, G\cup\text{Quad}_B\rangle)$. Since
 $$|Z| = \sum_{i\in\C I_1\cap\C I_2} (n_i + m_i - 1),$$
 the first statement is proved. For the final statement, observe that since $I$ is prime, we have $\dim I = \sum_{i\in\C I _1} n_i$, and since $J$ is prime, we have $\dim J = \sum_{i\in\C I_2}m_i$. If $\C I_1 = \C I_2 = [r]$ then we get
 $$|Z| = \sum_{i=1}^r (n_i + m_i - 1) = \sum_{i=1}^r n_i +  \sum_{i=1}^r m_i - r = \dim I + \dim J - |\C A|.$$
 Now since $I\times_{\mathcal{A}} J$ is prime, its dimension is equal to the size of any subset that is maximally independent modulo $\text{in}_\omega(\langle\text{Lift}\, F \cup\text{Lift}\, G\cup\text{Quad}_B\rangle)$.
 \end{proof}
 
 \begin{figure}[h!]
 \centering
  \resizebox{0.5\textwidth}{!}{
\begin{tikzpicture}

\fill[gray] (-1,-2) rectangle (0,4);
\fill[gray] (1,-2) rectangle (2,4);
\fill[gray] (3,-2) rectangle (4,4);

\fill[gray] (-2,2) rectangle (5,3);
\fill[gray] (-2,-1) rectangle (5,0);

\fill[black] (3,2) rectangle (4,3);
\fill[black] (-1,-1) rectangle (0,0);
\fill[black] (1,-1) rectangle (2,0);
\fill[black] (3,-1) rectangle (4,0);

\node at (-0.5, -2.5) {$j_i^1$};
\node at (1.5, -2.5) {$j_i^2$};
\node at (3.5, -2.5) {$j_i^3$};

\node at (-2.5, -0.5) {$k_i^1$};
\node at (-2.5, 2.5) {$k_i^2$};

\draw[step=1cm,black,very thin] (-2,-2) grid (5,4);

\end{tikzpicture}
}
\caption{Grid representing the generators $z_{jk}^i$ for a fixed $i\in [r]$. Columns shaded grey give monomials coming from the lift of the maximally independent set modulo $\text{in}_{\omega_1}(I)$, and rows shaded grey give monomials coming from the lift of the maximally independent set modulo $\text{in}_{\omega_2}(J)$. Cells shaded black represent the elements of the set $Z$ of degree $a_i$.}
\label{fig:tfp}
\end{figure}
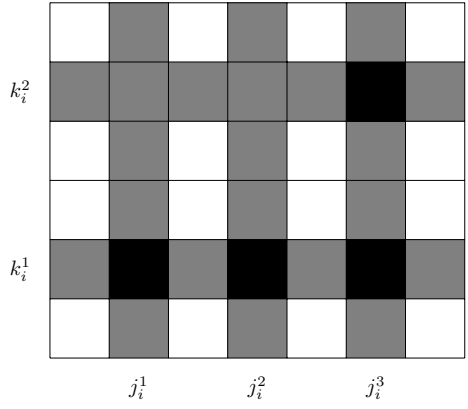
 
 \begin{remark}\label{rem:tfp}
 Observe that if both $I$ and $J$ are prime ideals and there exist maximally independent sets with $\C I_1 = \C I_2 = [r]$, then it is clear from the proof that $I\times_{\C A} J$ is also a prime ideal and there exists a maximally independent set modulo $\text{in}_\omega(I\times_{\C A} J)$ in $\{ z_{jk}^i\ |\ i \in [r],\, j\in [s_i],\, k\in [t_i]\}$ with at least one element for each upper index $i\in [r]$.
 \end{remark}
 
For the remainder of this section we will apply our results on toric fiber products to level-1 phylogenetic networks. Fix a group-based model $(G,B)$, and let $\C N$ be a level-1 phylogenetic network with a (directed) cut edge $e$. Then the operation of cutting $\C N$ at $e$ results in two smaller level-1 networks, that we denote $\C N_+$ and $\C N_-$. We denote by $e$ the new edge in both $\C N_+$ and $\C N_-$, and this edge inherits the direction from $\C N$. We assume that the network $\C N_+$  contains the leaves labelled $1, \ldots, n'$ for some $n' < n$, which are also leaves of $\C N$, and the new leaf, which we denote by $n_e$. Then $\C N_-$ contains the leaves labelled $n' + 1, \dots, n$, and the new leaf which we also denote by $n_e$.
 
The vanishing ideal $I_{\C N_+}$ is contained in the polynomial ring $R_+ = \mathbb{C}[ q_{g_1 \cdots g_{n'} g_{n_e}}\ |\ g_1 +\cdots + g_{n'} + g_{n_e} = 0]$, and $I_{\C N_-}$ is contained in $R_- = \mathbb{C}[ q_{g_{n_e} g_{n' + 1} \cdots g_n }\ |\  g_{n_e} + g_{n' + 1} +\cdots + g_n = 0]$. We give each polynomial ring the grading induced by $\text{deg}(q_{g_1\cdots g_N}) = E_{[\xi(e)]} \in \mathbb{Z}_{\ge 0}^{|B\cdot G|}$, where $\xi$ is the consistent edge labelling induced by the consistent leaf labelling $g_1,\ldots, g_N$, and $\{E_{[g]}\ |\ [g]\in B\cdot G\}$ is the standard basis of $\mathbb{Z}_{\ge 0}^{|B\cdot G|}$. Note that the set $\C A$ consisting of the image under $\text{deg}$ of the generators of $R_+$ and $R_-$ is given by the linearly independent set $\{ E_{[g]}\ |\ [g] \in B\cdot G\}$, and each element of this set is the image under $\deg$ of a generator of both $R_+$ and $R_-$. We assume that the edge $e$ is directed towards the leaves $1,\ldots, n'$, so that $\deg(q_{g_1\cdots g_n}) = E_{[g_1 + \cdots + g_{n'}]}$.
 
 We have a natural $\mathbb{C}$-algebra homomorphism
 \begin{equation}\label{eqn:naturalmap}
 \begin{aligned} 
 R &\to R_+ \otimes_{\mathbb{C}}R_- \\
 q_{g_1 \cdots g_n} &\mapsto q_{g_1 \cdots g_{n'} g_+} \otimes q_{g_{-} g_{n' + 1} \cdots g_n },
 \end{aligned}
 \end{equation}
 where $g_+ = -(g_1 + \cdots +g_{n'})$ and $g_{-} = -(g_{n'+1} + \cdots + g_{n})$. Note that $\deg(q_{g_1\cdots g_{n'}g_+}) = E_{[g_1 + \cdots +g_{n'}]}$ and $\deg(q_{g_-g_{n'+1}\cdots g_n}) = E_{[-(g_{n'+1} + \cdots + g_{n})]} = E_{[g_1 + \cdots +g_{n'}]}$. As in the proof of \cite[Proposition~3.2]{CHM}, the network parameterisation map $\phi_\C N$ factors through (\ref{eqn:naturalmap}), so $I_{\C N}$ is the toric fiber product $I_{\C N_+}\times_{\C A}I_{\C N_-}$.

 \begin{corollary}\label{cor:decomp}
 Fix a group-based model $(G,B)$. Let $\C N$ be a level-1 phylogenetic network with a cut edge $e$, and let $\C N_+$ and $\C N_-$ be the networks obtained by cutting $\C N$ at $e$. Then $\dim V^{(G,B)}_{\C N} = \dim V^{(G,B)}_{\C N_+} + \dim V^{(G,B)}_{\C N_-} - |B\cdot G|$.
 \end{corollary}
 \begin{proof}
 As described above, the ideal $I_{\C N}$ is the toric fiber product $I_{\C N_+}\times_{\C A}I_{\C N_-}$, so we apply Theorem \ref{thm:tfpdim}. Both $I_{\C N_+}$ and $I_{\C N_-}$ are prime ideals, so to prove the result, it is sufficient to show that for a phylogenetic network ideal $I$ there exists a weight vector $\omega$ and a set $U \subset \{q_{g_1\cdots g_n}\ |\ g_1 + \cdots +g_n = 0\}$ that is independent modulo $\text{in}_\omega(I)$ and that contains at least one element of degree $a_i$ for each $a_i \in \C A$. From Remark \ref{rem:tfp}, we need only consider phylogenetic networks that are either sunlet networks or trees. Furthermore, if $\C N$ is a sunlet network and $\C T$ is a tree obtained from $\C N$ be removing a reticulation edge, then $I_{\C N} \subset I_{\C T}$. It follows that if $U$ is independent modulo $\text{in}_\omega(I_{\C T})$ then $U$ is also independent modulo $\text{in}_\omega(I_{\C N})$, so in fact it is sufficient to show the result for any phylogenetic tree $\C T$. 
 
To show the result for a tree $\C T$, we make the further observation that if $\C T$ has an internal edge $e$, then $\C T$ is a toric fiber product of the two trees given by cutting $\C T $ at $e$. Thus in view of Remark \ref{rem:tfp} again, we need only consider claw trees. Since we are only considering binary phylogenetic trees, we need only consider the 3-claw tree $T_3$.

Fix a set of representatives $\C G \subset G$ of the $B$-orbits in $G$, let $\C T = T_3$, let $I = I_{\C T}$, and let the set of multidegrees be given by $\C A = \{E_g\ |\ g\in \C G\}$. Note that $0\in \C G$ and that $[0] = \{0\}$. We may assume, without loss of generality, that $\deg(q_{g'hk}) = E_{g}$, where $g' \in [g]$ for some $g\in\C G$. Recall that $I$ is given by the kernel of the map $\psi_{\C T}$ where

\begin{equation*}
\begin{aligned} 
\psi_{\C T} :\mathbb{C}[q_{ghk}\ |\ g+h+k=0] &\longrightarrow \mathbb{C}[a_i^{[g]}\ |\ g\in\C G,\, i = 1,2,3] \\
q_{ghk} &\longmapsto a_1^{[g]}a_2^{[h]}a_3^{[k]}.
\end{aligned} 
\end{equation*}

Let $U = \{ q_{g0(-g)}\ |\ g \in \C G\}$. It is clear that $U$ has exactly one element of each multidegree. Next, choose a term order on $\mathbb{C}[q_{ghk}\ |\ g + h + k = 0]$ such that $q_{ghk} < q_{g'h'k'}$ whenever $g\in\C G$ and $g' \not\in\C G$, and let $\omega$ be a a weight vector whose induced term order satisfies this. We claim that  $\mathbb{K}[U]\cap \text{in}_\omega(I) = \{0\}$. 

To prove the claim, we will show that for any  element $f \in I$, we have that $\text{in}_\omega(f)$ does not consist of a product of elements of $U$. Since $I$ is homogeneous and generated by binomials, we may assume that $f$ is a homogeneous binomial. Let $\C G' \subseteq \C G$ with $|\C G'| = n$ and suppose that we can write
$$f = \prod_{g\in \C G'} q_{g0(-g)} - m$$
for some other monomial $m$ of total degree $n$. Since $f \in \ker\psi_{\C T}$ we must have that 

$$\psi_{\C T}(m) = \psi_{\C T}\big(\prod_{g\in \C G'} q_{g0(-g)}\big) = \big(\prod_{g\in \C G'}a_1^{[g]}\big)\big(\prod_{g\in \C G'}a_3^{[-g]}\big)(a_2^{[0]})^n.$$
Now if $q_{g'h'k'}$ is a factor of $m$ then we must have $\psi_{\C T}(q_{g'h'k'}) = a_1^{[g]} a_2^{[0]} a_3 ^{[k]}$ for some $g, k\in\C G$. Thus, $h' \in [0]$ so $h' = 0$ and $g' \in [g]$, and since $g' +0 + k' = 0$ we must have $k' = -g'$. Now if $g' = g$ then $q_{g'h'k'} = q_{g0(-g)}$ appears as a factor in the first monomial of $f$. If this holds for all factors of $m$ then we have $f = 0$. If not, then for some factor $q_{g'0(-g')}$ we must have $g' \not\in\C G$, so we have $\text{in}_\omega(f) = m$.

 \end{proof} 
 
 \begin{remark}\label{rem:treedecomp}
 Notice that in the proof of Corollary \ref{cor:decomp}, we made no assumptions on the number of reticulation vertices of $\mathcal N$. Since a binary phylogenetic tree can be thought of as a phylogenetic network with no reticulation vertices, the result also holds for binary phylogenetic trees. Explicitly, we have that if $\C T$ is a binary phylogenetic tree with an interior edge $e$, with trees $\C T_+$ and $T_-$ obtained by cutting $e$, then we have
 $$\dim V^{(G,B)}_{\C T} = \dim V^{(G,B)}_{\C T_+} + \dim V^{(G,B)}_{\C T_-} - |B\cdot G|.$$
 \end{remark}

%
%
\section{Sunlet Networks and Trees}\label{sec:sunlet} 

If $\C N$ is a level-1 phylogenetic network, then $\C N$ can be decomposed along cut edges into a series of phylogenetic trees and sunlet networks. As shown in the previous section, the ideal structure of the corresponding varieties is given by the toric fiber product. It therefore remains for us give dimension results for the varieties corresponding to trees and sunlet networks. For an unrooted phylogenetic tree $\C T$, the dimension of the variety $V_{\C T}^{(G,B)}$ is well known. We give a proof using the dimension result of the previous section.
 \begin{lemma}\label{lemma:treedim}
  If $\C T$ is a binary phylogenetic tree with $m$ edges and no degree-2 vertices under a group-based evolutionary model $(G,B)$, then the affine dimension of $V^{(G,B)}_{\C T}$ is given by 
   \[ \dim V^{(G,B)}_{\C T} = lm+1. \] 
 \end{lemma}
 
 \begin{proof}
 Denote by $t$ the number of interior edges of $\C T$. If $t=0$ then $\C T$ is the 3-claw tree. This has dimension $3l + 1$ by \cite[Proposition~5.2]{MRC2}, so the proposition is true in this case.
 
 Now suppose $\C T$ is a binary phylogenetic tree with $m$ edges and $t > 0$ interior edges. Let $e$ be an interior edge and let $\C T_+$ and $\C T_-$ be the trees obtained by cutting at $e$. If $m_+$ and $m_-$ are the number of edges of $\C T_+$ and $\C T_-$ respectively, we have $m = m_+ + m_- - 1$.  Furthermore, the number of interior edges of $\C T_+$ and of $\C T_-$ is less than $t$, so by induction we have $\dim V^{(G,B)}_{\C T_+} = lm_+ + 1$ and $\dim V^{(G,B)}_{\C T_-} = lm_- + 1$. It follows from Remark \ref{rem:treedecomp} that
 $$ \dim V^{(G,B)}_{\C T} = (lm_+ + 1) + (lm_- + 1) - (l+1) = lm + 1.$$
 \end{proof}
 Observe that one could extend the above proof to give the analogous dimension result for any phylogenetic tree with no degree $2$ vertices. To do so, the base-case for the induction must be extended to cover all claw trees $T_n$ with $n \geq 3$. That is, one must show that for each $T_n$ with corresponding ideal $I_n$, there exists a maximal independent set modulo $\text{in}_\omega(I_n)$ that contains at least one element of multidegree $a_i$ for each $a_i\in\C A$, as in the proof of Corollary \ref{cor:decomp}. 
 
 The remainder of this section is dedicated to giving the dimension of the varieties corresponding to sunlet networks. As we have already seen, the variety associated to a phylogenetic network $\C N$ is equal to the variety associated to the corresponding contracted semi-directed network, so from this point onwards we will only consider contracted semi-directed networks.  First we will give an upper bound on the dimension.
 
 
   \begin{proposition}\label{prop:upper}
  If $\C N$ is a contracted semi-directed phylogenetic network with only disjoint cycles and with $m$ edges then
   \[ \dim V_{\C N}^G \leq lm + 1. \]
 \end{proposition}
 \begin{proof}
 Let $c$ denote the number of cycles in $\C N$. The affine variety $V^G_{\C N}$ is parameterized by $(l+1)m$ parameters, but the map is multihomogeneous. It is linear in the set of parameters for each non-reticulation edge, and in the union of the parameters for the two reticulation edges of each cycle.  Thus we may think of the parameterization map as a projective map
$$ \mathbb{P}^l\times\cdots\times\mathbb{P}^l \times \mathbb{P}^{2l+1} \times\cdots\times \mathbb{P}^{2l+1} \dashrightarrow \mathbb{P}^{(l+1)^{n-1}}, $$
where $\mathbb{P}^l$ appears $m-2c$ times (once for each non-reticulation edge), and $\mathbb{P}^{2l+1}$ appears $c$ times (once for each cycle). We use a dashed arrow to indicate that in order for the map to be well-defined we may need to take a subset of the domain. It follows that  $V^G_{\C N}$ has projective dimension at most $lm + c$, and thus its affine dimension is at most $lm +c + 1$.

Now consider $v = \phi_{\C N}(w) \in \mathbb{C}^{|G|^{n-1}}$, where $w\in\mathbb{C}^{m(l+1)}$. For each pair of reticulation edges $e_1,e_2$, a consistent leaf labelling of $\C N$ assigns both edges the same label.  For each consistent leaf labelling of $\C N$ in which they are labelled 0, the edges along the cycle all receive the same labels in both trees, so the coordinate of $v$ corresponding to the consistent leaf labelling has a factor of $w_{e_1}^0 + w_{e_2}^0$.  For every consistent leaf labelling in which they are not labelled 0, the coordinate does not depend on $w_{e_1}^0$ or $w_{e_2}^0$.  Therefore the map depends only on the sum $w_{e_1}^0 + w_{e_2}^0$. This reduces the number of independent parameters by $c$, so the affine dimension of $V^G_{\C N}$ is at most $lm + 1$.
 \end{proof}
  
 \subsection{General group-based models}
 
First, we restrict our attention to sunlet networks under general group-based models of evolution, i.e., those where the group $B$ consists only of the identity automorphism. We will deal with the case $G=\mathbb{Z}/2\mathbb{Z}$ separately.
 
  \begin{proposition}\label{prop:4SN}
Let $\C N$ be the $n$-sunlet network with $n \geq 4$ and let $G$ be an abelian group with $|G| = l + 1 > 2$. Then
  \[ \dim V^G_{\C N} = l(2n -1) + 1. \]
 \end{proposition}
 \begin{proof}
Using Lemma \ref{lemma:contractedNetwork}, we may replace $\C N$ by its contraction. This network has  $m = 2n - 1$ edges, so by Proposition \ref{prop:upper} we have $\dim V^G_{\C N} \leq l(2n -1) + 1$. 

Label the edges and vertices as in Figure \ref{fig:sunlet}, and let $\C T_1$ and $\C T_2$ be the two trees got by removing the edges $e_{n+1}$ and ${e_1}$ respectively. The parameterization map for $\C N$ is given by
\begin{equation} \label{eqn:sunletParam}
q_{g_1g_2\cdots g_n} = a_2^{g_2}\cdots a_n^{g_n}\big(a_1^{g_1}a_{n+2}^{g_1 + g_2}\cdots a_{2n-1}^{g_n} + a_{n+1}^{g_1} a_{n+2}^{g_2}\cdots a_{2n-1}^{g_1+g_n}\big),
\end{equation}
where $g_1, \ldots, g_n$ is a consistent leaf labelling, the first monomial corresponds to $\C T_1$, and the second monomial corresponds to $\C T_2$. With notation as in Section \ref{subsec:tropical}, our aim will be to find $\lambda$ that maximises  $\rank_\mathbb{R} A_\lambda$, which by Lemma \ref{lemma:draisma} gives a lower bound on $\dim V^G_{\C N}$.
  
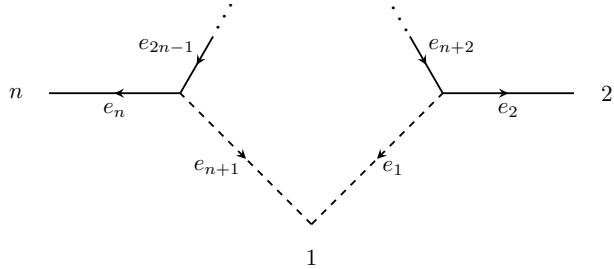
\begin{figure}[h!]
 \centering
 \resizebox{0.65\textwidth}{!}{
\begin{tikzpicture}
[every node/.style={inner sep=0pt},
                    every path/.style={thick},   
decoration={markings, 
    mark= at position 0.5 with {\arrowreversed{stealth}}
    }
]
                    
\draw [postaction={decorate}] (0,2) -- (2,2) node[midway,below=4] {$e_{n}$};
\draw [postaction={decorate}] (2,2)  -- (2.5 ,2.866)  node[midway,above=8,left=0.5] {$e_{2n-1}$};
\draw [postaction={decorate}] (8,2) -- (6,2) node[midway,below=4] {$e_{2}$};
\draw [postaction={decorate}] (6,2)  -- (5.5 , 2.866)  node[midway,above=8,right=0.5] {$e_{n+2}$};
\draw[dashed,postaction={decorate}] (4,0)  -- (2,2) node[midway,below=4,left=2] {$e_{n+1}$};
\draw[dashed,postaction={decorate}] (4,0)  -- (6,2) node[midway,below=4,right=2] {$e_1$};
\node[rotate=60] (v4) at (2.7,3.2) {{\Large $\cdots$}};
\node[rotate=-60] (v4) at (5.35,3.1) {{\Large $\cdots$}};

\node (v3) at (4,-0.5) {$1$};
\node (v3) at (8.5,2) {$2$};
\node (v3) at (-0.5,2) {$n$};

\end{tikzpicture}
}
\caption{A contracted $n$-sunlet network with $n \geq 4$. Arrows indicate the orientation used for assigning a consistent edge labelling from a consistent leaf labelling.}
\label{fig:sunlet}
\end{figure}
 
Let $\{ E_i^g\ |\ g\in G, i = 1,\ldots, 2n-1 \}$ be the standard basis of $\mathbb{R}^{m(l+1)}$, indexed by the edges of $\C N$ and elements of $G$, and consider the dual vector space $V=(\mathbb{R}^{m(l+1)})^*$ with the dual basis. Choose $\lambda \in V$ such that $\lambda_{n+2}^0 = -2$, $\lambda_{n+1}^g = 1 $ for all $g\in G$, and all other entries are 0. Let $g_1, \ldots, g_n$ be a consistent leaf labelling of $\C N$. Then the corresponding column of $A_\lambda$ has the following properties: 
 \begin{itemize}
 \item If $g_1 = 0$, then the monomial from $\C T_1$ is chosen.
 \item If $g_1 \neq 0$ and $g_2 = 0$, then the monomial from $\C T_2$ is chosen.
 \item In all other cases the monomial from $\C T_1$ is chosen.
 \end{itemize}
We will show that $\rank A_\lambda \geq l(2n -1) + 1$ to give the lower bound.
 
Consider the submatrix given by consistent leaf labellings where $g_1 = 0$, so that each column is an exponent vector coming from a monomial in $\C T_1$. Perform column operations on $A_\lambda$ so that the first $(l+1)^{n-2}$ columns are given by this submatrix.  Let $\C S$ be the tree with $n-1$ leaves obtained from $\C N$ by deleting the reticulation vertex.  The consistent leaf labellings of $\C N$ in which $g_1 = 0$ give all of the consistent leaf labellings of $\C S$.  The tropical tree model for $\C S$ has dimension $l(2n-5) + 1$ by Lemma \ref{lemma:treedim}, so the submatrix of $A_\lambda$ consisting only of the columns where $g_1 = 0$ has rank $l(2n-5)+1$ (note that since $\C S$ has a monomial parameterization, for all choices of $\lambda$ we have that $\dim S$ is equal to the rank of this submatrix).

We make the following observations about this submatrix.  First, since the monomial from $\C T_1$ is always chosen, the entries corresponding to the parameters $a_{n+1}^g$ are $0$ for all $g \in G$. Similarly, for the edge $e_1$, only the parameter $a_1^0$ appears in the parameterization of $q_{g_1 g_2\cdots g_n}$, so the entries corresponding to the parameters $a_{1}^g$ are $0$ for all $g \in G$ except $g=0$. Next, observe that in this submatrix, the row corresponding to the parameter $a_2^g$ is equal to the row corresponding to the parameter $a_{n+2}^g$ for all $g \in G$ and similarly the row corresponding to the parameter $a_n^g$ is equal to the row corresponding to the parameter $a_{2n-1}^g$ for all $g \in G$. This is because the label of e.g. the edge $e_{n+2}$ is $g_1 + g_2 = 0 + g_2 = g_2$, which is also the label of the edge $e_2$. We perform row operations on $A_\lambda$ so that for each $g \in G$,  the first $(l+1)^{n-2}$ entries of the rows corresponding to the parameters  $a_{2}^g$  and  $a_{n}^g$ are zero, by subtracting the rows $a_{n+2}^g$ and $a_{2n-1}^g$ respectively. Now we perform further row operations to swap rows and obtain a matrix of the following form, where the upper left block is a $(4l + 3)\times (l+1)^{n-2}$ matrix consisting of zeros,

\begin{equation*}
	A_\lambda = \left[
	\begin{array}{c|c}
	0& B\\ \hline
	A_\lambda '& *
	\end{array}\right],
  \end{equation*}
and $\rank A_\lambda ' = l(2n -5) + 1$. It follows that $\rank A_\lambda \geq l(2n -5) + 1 + \rank B$, so it is sufficient to show that $\rank B \geq 4l$.

The columns of $B$ correspond to consistent leaf labellings $g_1, \ldots, g_n$ with $g_1 \neq 0$. Recall that $\lambda$ was such that if $g_2 = 0$ then the monomial from $\C T_2$ is chosen, and otherwise the monomial from $\C T_1$ is chosen. The rows of $B$ correspond to the parameters $a_1^g$ for $g\neq 0$, and $a_{n+1}^g, a_{2}^g$ , and $a_{n}^g$ for all $g\in G$. However, we performed row operations on the rows corresponding to $a_{2}^g$ and $a_{n}^g$, so for each column of $B$ the coefficient of the standard basis vector $E_2^g$ is given by the exponent of $a_2^g$ minus the exponent of $a_{n+2}^g$, and the coefficient of $E_n^g$ is given by the exponent of $a_n^g$ minus the exponent of $a_{2n-1}^g$ in the corresponding monomial from the parameterization (\ref{eqn:sunletParam}). Thus the columns of B are given by 
\begin{equation}\label{eqn:B1}
(E_n^{g_n} - E_n^{g_1 + g_n}) + E_{n+1}^{g_1},
\end{equation}
if $g_2 = 0$ (so the monomial comes from $\C T_2$), and 
\begin{equation}\label{eqn:B2}
(E_2^{g_2} -E_2^{g_1 + g_2}) + E_{1}^{g_1},
\end{equation}
otherwise (so the monomial comes from $\C T_1$), where $g_1, \ldots, g_n$  is a consistent leaf labelling with $g_1 \neq 0$. Note that since $n \geq 4$, we can find a consistent leaf labelling $g_1, \ldots, g_n$ for any choice of $g_1$, $g_2$, $g_n$. Denote by $X_1$ the vector space spanned by all the vectors of the form in equation (\ref{eqn:B1}). We have
$$\sum_{g\in G} \big((E_n^{g} - E_n^{g_1 + g}) + E_{n+1}^{g_1}\big) = (l+1)E_{n+1}^{g_1} \in X_1,$$
so $E_{n+1}^{g_1} \in X_1$ for all $g_1 \neq 0$. It follows immediately that for a fixed $g_n \in G$, we have $E_n^{g_n} - E_n^{g_1 + g_n}\in X_1$ for all $g_1 \neq 0$, and thus $\dim X_1 \geq 2l$.

Next denote by $X_2$ the vector space spanned by all the vectors of the form in equation (\ref{eqn:B2}). Using that $\sum_{g\in G} E_2^g - E_2^{g_1 + g} = 0$, we see that
$$\sum_{g \in G\setminus\{0\}}\big((E_2^g - E_2^{g_1 +g}) + E_1^{g_1}\big) = lE_1^{g_1} + E_2^{g_1} - E_2^0 \in X_2,$$
for each $g_1\neq 0$. Now fix $g\in G\setminus\{0\}$ and let $g_2 = g$, and $g_1 = -g$, so that $E_1^{-g} + E_2^g - E_2^0 \in X_2$. Then we have
$$(lE_1^g + E_2^g - E_2^0) - (E_1^{-g} + E_2^g - E_2^0) = lE_1^g - E_1^{-g} \in X_2.$$
Now if $g = -g$ then we have $E_1^g \in X_2$. If not, by swapping $g_1$ and $g_2$ we have $lE_1^{-g} - E_1^g \in X_2$, so $(l-1)E_1^g - (l-1)E_1^{-g} \in X_2$. Then $lE_1^g - lE_1^{-g} \in X_2$, and subtracting $lE_1^{-g} - E_1^g$ gives $E_1^g \in X_2$, for all $g\in G\setminus\{0\}$. As before, it follows that for a fixed $g_2$ we have $E_2^{g_2} - E_2^{g_1+g_2}\in X_2$ for all $g_1 \in G\setminus\{0\}$, so $\dim X_2 \geq 2l$. It follows that $\rank B \geq 4l$. 
  \end{proof}
  
  Next we deal with the case $G = \mathbb{Z}/2\mathbb{Z}$. The expected dimension for $n$-sunlets here is $2n$. However, if $n = 3$ then we only have $4 < 2n$ consistent leaf labellings of $\C N$, so in this case the expected dimension cannot be reached. When $n=4$ we have $8 = 2n$ consistent leaf labellings, however, in this case $\dim V_{\C N}^{G} = 7$. This can be shown by direct computation.
 
 \begin{proposition}\label{prop:5SN}
Let $\C N$ be the $n$-sunlet network with $n \geq 5$ and let $G = \mathbb{Z}/2\mathbb{Z}$. Then
  \[ \dim V^G_{\C N} = 2n.\]
 \end{proposition}
 \begin{proof}
As before, Proposition \ref{prop:upper} gives the upper bound. Label the edges and vertices as in Figure \ref{fig:sunlet5}, and let $\C T_1$ and $\C T_2$ be the two trees got by removing the edges $e_{n+1}$ and ${e_1}$ respectively. Observe that we have at least one edge on the cycle, e.g. $e_{n+3}$, that is not adjacent to either reticulation edge. We proceed as in Proposition \ref{prop:4SN}, this time choosing $\lambda \in \mathbb{R}^{2m}$ such that $\lambda_{n+1}^0 = \lambda_{n+1}^1 = 1$, $\lambda_{n+3}^0 = 2$. and all other entries are 0. Let $g_1, \ldots, g_n$ be a consistent leaf labelling of $\C N$. Then the corresponding column of $A_\lambda$ has the following properties: 
 \begin{itemize}
 \item If $g_1 = 0$, then the monomial from $\C T_1$ is chosen.
 \item If $g_1 = 1$ and $g_2 + g_3 = 1$, then the monomial from $\C T_2$ is chosen.
 \item If $g_1 = 1$ and $g_2 + g_3 = 0$, then monomial from $\C T_1$ is chosen.
 \end{itemize}
 
  \begin{figure}[h!]
 \centering
  \resizebox{0.65\textwidth}{!}{
\begin{tikzpicture}
[every node/.style={inner sep=0pt},
                    every path/.style={thick},   
decoration={markings, 
    mark= at position 0.5 with {\arrowreversed{stealth}}
    }
]
                    
\draw [postaction={decorate}] (0,2) -- (2,2) node[midway,below=4] {$e_{n}$};
\draw [postaction={decorate}] (2,2)  -- (2.3 ,2.866)  node[midway,above=8,left=0.5] {$e_{2n-1}$};
\draw [postaction={decorate}] (8,2) -- (6,2) node[midway,below=4] {$e_{2}$};
\draw [postaction={decorate}] (6,2)  -- (5.8 , 3.5)  node[midway,above=5,right=2.0] {$e_{n+2}$};
\draw [postaction={decorate}]  (5.8 , 3.5) -- (5,4) node[midway,above=13] {$e_{n+3}$};
\draw [postaction={decorate}] (7, 4.5) -- (5.8 , 3.5) node[midway,below=8,right=0.5] {$e_{3}$};
\draw[dashed,postaction={decorate}] (4,0)  -- (2,2) node[midway,below=4,left=2] {$e_{n+1}$};
\draw[dashed,postaction={decorate}] (4,0)  -- (6,2) node[midway,below=4,right=2] {$e_1$};                  
\node[rotate=70] (v4) at (2.45,3.215) {{\Large $\cdots$}};
\node[rotate=-30] (v4) at (4.8,4.1) {{\Large $\cdots$}};

\node (v3) at (4,-0.5) {$1$};
\node (v3) at (8.5,2) {$2$};
\node (v3) at (7.3,4.7) {$3$};
\node (v3) at (-0.5,2) {$n$};

\end{tikzpicture}
}
\caption{A contracted $n$-sunlet network with $n \geq 5$. Arrows indicate the orientation used for assigning a consistent edge labelling from a consistent leaf labelling.}
\label{fig:sunlet5}
\end{figure}
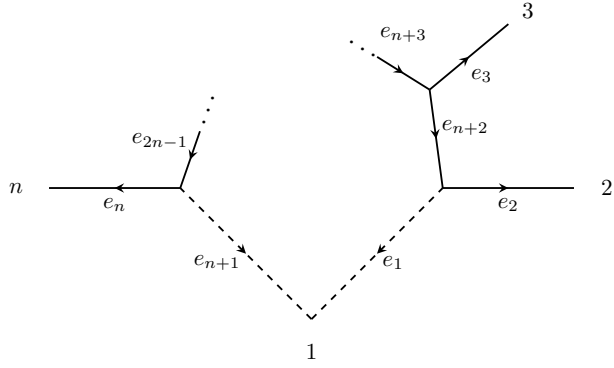

As in the proof of Proposition \ref{prop:4SN}, we perform column operations so that the first $2^{n-2}$ columns are indexed by consistent leaf labellings where $g_1 = 0$, and each column is an exponent vector coming from the corresponding $\C T_1$ monomial. The submatrix consisting of these columns has rank $2n - 4$, and we perform the same row operations as before to give the block triangular matrix
\begin{equation*}
	A_\lambda = \left[
	\begin{array}{c|c}
	0& B\\ \hline
	A_\lambda '& *
	\end{array}\right],
  \end{equation*}
 where the submatrix $B$ is given by rows corresponding to the parameters $a_1^1, a_2^0, a_2^1$, $a_n^0, a_n ^1, a_{n+1}^0$, and $a_{n+1}^1$. However, we performed row operations on the rows corresponding to $a_{2}^g$ and $a_{n}^g$, so for each column of $B$ the coefficient of $E_2^g$ is given be the exponent of $a_2^g$ minus the exponent of $a_{n+2}^g$, and the coefficient of $E_n^g$ is given by the exponent of $a_n^g$ minus the exponent of $a_{2n-1}^g$ for $g = 0,1$. Consider the following columns of $B$. For a consistent leaf labelling with $g_1 = g_n = 1$ and $g_2 = g_3 = 0$, the monomial from $\C T_1$ is chosen, so the labels assigned to $a_n$ and $a_{2n-1}$ are equal, and the labels assigned to $a_2$ and $a_{n+2}$ are not equal. Thus the column is given by  
 $$E_1^1 + E_2^0 - E_2^1.$$
 Next for a consistent leaf labelling with $g_1 = g_2 = g_3 = g_n = 1$, the monomial from $\C T_1$ is chosen so the column is given by
 $$E_1^1 - E_2^0 + E_2^1.$$
 For $g_1 = g_3 = 1$ and $g_2 = g_n = 0$, the monomial from $\C T_2$ is chosen so the column is given by
 $$E_{n+1}^1 + E_n^0 - E_n^1.$$
 Finally, for a consistent leaf labelling with  $g_1 = g_3 = g_n = 1$ and $g_2 = 0$, the monomial from $\C T_2$ is chosen so the column is given by
 $$ E_{n+1}^1 - E_n^0 + E_n^1.$$
These vectors are linearly independent, so $\rank B \geq 4$ and the result follows.
 \end{proof}
 
 \subsection{Group-based models}
 
 In this section, we use our results on general group-based models to obtain the result for all group-based models, following the method of \cite[Lemma~4.2]{MRC2}. Throughout, let $\C N$ be the contracted $n$-sunlet network, so the number of edges $m$ is equal to $2n-1$. Let $G$ be a finite abelian group, and let $B$ be a subgroup of the automorphism group $\rm{Aut}(G)$ with $|B\cdot G| = l +1$. Let $(\mathbb{R}^{|G|m})^*$ have standard basis elements $\varepsilon_{e}^g$ where $g\in G$ and $e \in \C E(\C N)$. Next, pick representatives $g_0 = 0, g_1, \ldots, g_l$ in $G$ for each $B$-orbit, and let $(\mathbb{R}^{(l+1)m})^*$ have standard basis elements $\varepsilon_{e}^{[g_i]}$ for $i=0,\ldots,l$ and $e \in \C E(\C N)$.
 
 Let $p:(\mathbb{R}^{|G|m})^*\longrightarrow(\mathbb{R}^{(l+1)m})^*$ be the map that sums coefficients of the unit vectors for each orbit, i.e.
 $$ \sum_{e \in \C E(\C N)} \sum_{g \in G} c_e^g \varepsilon_e^g \longmapsto  \sum_{e \in \C E(\C N)} \sum_{i = 0}^l (\sum_{g\in [g_i]} c_e^g) \varepsilon_e^{[g_i]},$$
 where  $c_e^g \in \mathbb{R}$. It is clear that $p$ is a surjective, linear map, so $\dim\ker p = (|G| - l - 1)m$. Now consider the parameterizations of $V^{G}_{\C N}$ and $V^{(G,B)}_{\C N}$. For a fixed consistent leaf labelling $\xi$, let $\alpha_1$ and $\alpha_2$ be the exponent vectors of the monomials corresponding to $\C T_1$ and $\C T_2$ respectively, in the parameterization of $V^{G}_{\C N}$. Similarly let $\alpha_1^\prime$ and $\alpha_2^\prime$ be the corresponding monomials in the parameterization of $V^{(G,B)}_{\C N}$. Then $p(\alpha_i) = \alpha_i^\prime$ for $i = 1, 2$. Furthermore, observe that if $\lambda \in \mathbb{R}^{|G|m}$ is such that $\lambda_e^g = \lambda_e^h$ whenever $g$ and $h$ are in the same $B$ orbit for all edges $e\in\C E(\C N)$, then there exists $\lambda^\prime \in \mathbb{R}^{(l+1)m}$ satisfying $\lambda = \lambda ' \circ p$ (where we are considering $\lambda^\prime$ as an element of the dual space of $(\mathbb{R}^{(l+1)m})^*$).
 
 \begin{proposition}\label{prop:reduction}
With the notation as above, let $\lambda \in \mathbb{R}^{|G|m}$ be such that $\lambda_e^g = \lambda_e^h$ whenever $g$ and $h$ are in the same $B$-orbit, for all $e \in \C E_{\C N}$. Then there exists $\lambda^\prime \in \mathbb{R}^{(l+1)m}$ such that
\[ p \circ A_\lambda = A_{\lambda^\prime}. \]
 \end{proposition}
 \begin{proof}
 First observe that for any $\alpha \in (\mathbb{R}^{|G|m})^*$ we have 
 $$\langle \lambda, \alpha \rangle = \lambda(\alpha) = \lambda^\prime \circ p(\alpha) = \langle \lambda^\prime, p(\alpha)\rangle.$$
Now consider the polynomials of the parameterizations of $V^{G}_{\C N}$ and $V^{(G,B)}_{\C N}$, for a consistent leaf labelling $\xi$. Let $\alpha_1$ and $\alpha_2$ be as above, and suppose that $\langle \lambda, \alpha_1\rangle < \langle \lambda, \alpha_2\rangle$. Then $\langle \lambda^\prime, p(\alpha_1)\rangle < \langle \lambda^\prime, p(\alpha_2)\rangle$, so both $\lambda$ and $\lambda^\prime$ pick the monomial corresponding to $\C T_1$.  Since $\alpha_1^\prime = p(\alpha_1)$, the result follows.
\end{proof}

Note that Proposition \ref{prop:reduction} is easily generalizable to level-1 phylogenetic networks.
 
 \begin{corollary}\label{cor:non-general}
 
 Let $\C N$ be the $n$-sunlet network with $n \geq 4$, let $G$ be a finite abelian group, and let $B$ be a non-trivial subgroup of the automorphism group $\rm{Aut}(G)$, with $|B\cdot G| = l +1$. Then
  \[ \dim V^{(G,B)}_{\C N} = l(2n -1) + 1. \]
  
 \end{corollary}
 \begin{proof}
As in the case for general group-based models, the upper bound is given by Proposition \ref{prop:upper}. For the lower bound, first observe that since $B$ is a non-trivial subgroup, we must have $|G| > 2$. Next observe that the vector $\lambda$ chosen in the proof of Proposition \ref{prop:4SN} satisfies the condition in Proposition \ref{prop:reduction}, so using Proposition \ref{prop:reduction} (and  Lemma \ref{lemma:draisma}) there exists some $\lambda'$ such that
$$\dim V^{(G,B)}_{\C N}  \geq \rank_\mathbb{R} A_{\lambda^\prime} = \rank_\mathbb{R}(p \circ A_{\lambda^\prime}).$$
Finally, since $p$ is a surjective linear map with kernel of dimension $ (|G| - l - 1)m$, we have 
$$ \rank_\mathbb{R} A_{\lambda^\prime} \geq (|G| - 1)m + 1 - (|G| - l - 1)m = lm + 1.$$
 \end{proof}
 
 We summarise our results on sunlet networks in a single theorem. Note that the final two cases are given by direct computation.
 
 \begin{theorem}\label{thm:sunlet}
 Let $\C N$ be a sunlet network with $n$ leaves. Let $G$ be a finite abelian group and let $B$ be a subgroup of $\Aut(G)$. Denote by $l+1$ the number of $B$-orbits in $G$. Then $\dim V_{\C N}^{(G,B)}$ is given in the following cases.
 \begin{itemize}
 \item If $n\geq 4$ and $|G| > 2$ then $\dim V_{\C N}^{(G,B)} = l(2n-1) + 1$.
 \item If $n\geq 5$ and $G =\mathbb{Z}/2\mathbb{Z}$ so that $B = \{\rm{id}\}$ then $\dim V_{\C N}^{\mathbb{Z}/2\mathbb{Z}} = 2n$.
 \item If $n=4$ then $\dim V_{\C N}^{\mathbb{Z}/2\mathbb{Z}} = 7$.
  \item If $n=3$ then $\dim V_{\C N}^{\mathbb{Z}/2\mathbb{Z}} = 4$.
  \end{itemize} \qed
 \end{theorem}
 
%
%
  \section{Proof of Theorems \ref{thm:main} and \ref{thm:zmod2z}}
  
  We are now able to give simple inductive proofs of Theorems \ref{thm:main} and \ref{thm:zmod2z}. Below we give only the proof of Theorem \ref{thm:main}. The proof of Theorem  \ref{thm:zmod2z} is almost identical, and is left to the reader with the aid of Table \ref{tab:def}.
  
  \begin{proof}[Proof of Theorem~\ref{thm:main}]
  We will prove the result using induction on the number of cut edges of a level-1, triangle-free phylogenetic network $\C N$. For the case when there are no cut edges, we must have that $\C N$ is either the 3-claw tree, in which case the dimension of $V_{\C N}^{(G,B)}$ is equal to $lm + 1$ by Lemma \ref{lemma:treedim}, or $\C N$ is an $n$-sunlet network with $n\geq 4$, in which case the dimension is $l(2n - 1) + 1$ by Theorem \ref{thm:sunlet}. In both cases the result holds.
  
  Now suppose that $\C N$ is a level-1, triangle-free phylogenetic network with a cut edge $e$, and $m$ edges and $c$ cycles. Let $\C N_1$ and $\C N_2$ be the networks obtained by cutting at $e$, and let $m_i$ and $c_i$ denote the number of edges and cycles in $\C N_i$ respectively for $i=1,2$. Since the number of cut edges in $\C N_1$ and $\C N_2$ must be fewer than the number of cut edges in $\C N$,  by induction we have $\dim V_{\C N_i}^{(G,B)} = l(m_i - c_i) + 1$ for $i = 1,2$. By Corollary \ref{cor:decomp} we have
\begin{equation*}
 \begin{aligned}
 \dim V_{\C N}^{(G,B)} 	&= \dim V_{\C N_1}^{(G,B)}  + \dim V_{\C N_2}^{(G,B)}  - (l+1) \\
 					&= l(m_1 + m_2 - c_1 - c_2) + 2 - (l+1) \\
					&= l(m - c ) + 1,
 \end{aligned}
 \end{equation*}
where $m_1 + m_2 = m+1$ and $c_1 + c_2 = c$.
  \end{proof}
  
%
%
  \section{Application to Identifiability}
 
 In this section we apply Theorems \ref{thm:main} and \ref{thm:zmod2z} to give some immediate identifiability results. Throughout, fix an abelian group $G$ and subgroup $B$ of $\Aut(G)$, and let $l+1$ be the number of orbits in $B\cdot G$. First, we extend the definition of \emph{distinguishibility} from \cite{GL} to all group-based models of evolution
 
 \begin{definition}
 Let $(G,B)$ be a group-based model of evolution. Two distinct $n$-leaf networks $\C N_1$ and $\C N_2$ are \emph{distinguishable over $(G,B)$} if $V_{\C N_1}^{(G,B)} \not\subseteq V_{\C N_2}^{(G,B)}$ and $V_{\C N_2}^{(G,B)} \not\subseteq V_{\C N_1}^{(G,B)}$.
 \end{definition}
 
 When $G$ and $B$ are clear, we will simply say that $\C N_1$ and $\C N_2$ are distinguishable. Observe that if $V_{\C N_1}^{(G,B)}$ and $V_{\C N_2}^{(G,B)}$ are irreducible varieties of equal dimension, then in order to determine whether $\C N_1$ and $\C N_2$ are distinguishable it is sufficient to show that either  $V_{\C N_1}^{(G,B)} \not\subseteq V_{\C N_2}^{(G,B)}$ or $V_{\C N_2}^{(G,B)} \not\subseteq V_{\C N_1}^{(G,B)}$. One of the key results we will use to show identifiability is the following.

 \begin{lemma}[\cite{gross2021distinguishing} Lemma~3]\label{lem:subnet}
 Let $\C N_1$ and $\C N_2$ be $n$-leaf networks.  If for some $A \subseteq [n]$, we have that  $V^{(G,B)}_{\C N_1|_A} \not\subseteq V^{(G,B)}_{\C N_2|_A}$, then  $V^{(G,B)}_{\C N_1} \not\subseteq V^{(G,B)}_{\C N_2}$.
\end{lemma}

 \begin{corollary}\label{cor:subnet}
  Let $\C N_1$ and $\C N_2$ be $n$-leaf networks with $\dim V^{(G,B)}_{\C N_1} = \dim V^{(G,B)}_{\C N_2}$.  If for some $A \subseteq [n]$ we have $V^{(G,B)}_{\C N_1|_A} \not\subseteq V^{(G,B)}_{\C N_2|_A}$, then $\C N_1$ and $\C N_2$ are distinguishable over $(G,B)$.
 \end{corollary}
  \begin{proof}
  By Lemma \ref{lem:subnet}, $V^{(G,B)}_{\C N_1} \not\subseteq V^{(G,B)}_{\C N_2}$.  Since they are irreducible varieties of the same dimension, they are distinguishable.
 \end{proof}
 
 We will use Corollary \ref{cor:subnet} in conjunction with the following dimension results.
 
 \begin{lemma}\label{lem:level1dim}
 Let $\C N_1$ and $\C N_2$ be $n$-leaf, level-1 phylogenetic networks, both with exactly $c$ cycles, where each cycle has length at least $4$ when $|G|> 2$ and at least $5$ when $G=\mathbb{Z}/2\mathbb{Z}$. Then $\dim V^{(G,B)}_{\C N_1} = \dim V^{(G,B)}_{\C N_2}$.
 \end{lemma}
 \begin{proof}
 Observe that $\C N_1$ and $\C N_2$ have the same number of edges. To see this, suppose that $\C N_1$ and $\C N_2$ have $m_1$ and $m_2$ edges respectively. Then the corresponding contracted networks $\C N_1'$ and $\C N_2'$ have $m_1 - c$ and $m_2 -c$ edges, since for each reticulation vertex the outgoing edge is removed. Next for each of the $c$ reticulation vertices $v_1,\ldots, v_c$ in $\C N_1'$ arbitrarily pick a reticulation edge $(u_i,v_i)$ and remove it. After removal, the vertex $u_i$ has degree 2 and can be suppressed. The result is an unrooted binary phylogenetic tree on $n$ leaves with $m_1 - 3c$ edges. Performing the same operations on $\C N_2'$ we also obtain a (possibly different) unrooted binary phylogenetic tree on $n$ leaves with $m_2-3c$ edges. Since all unrooted binary phylogenetic trees on $n$ leaves have $2n-3$ edges, we have that $m_1 = m_2$. Now since both $\C N_1$ and $\C N_2$ have exactly $c$ cycles, the result follows from Theorems \ref{thm:main} and \ref{thm:zmod2z}.
 \end{proof}
 \begin{remark}\label{rem:edgenum}
 From the proof of Lemma \ref{lem:level1dim} it is clear that the number of edges of an unrooted level-1 phylogenetic network on $n$ leaves with $c$ cycles is $2n-3 + 3c$
 \end{remark}
 
For the remaining results in this section we will need to use the fact that binary phylogenetic trees with group-based models of evolution are distinguishable. This result is well-known in the community, but we give a direct proof here for completeness.

\begin{lemma}\label{lem:treedis}
Let $(G,B)$ be a group-based model of evolution, and let $\C T_1$ and $\C T_2$ be two distinct $n$-leaf, unrooted, binary phylogenetic trees. Then $\C T_1$ and $\C T_2$ are distinguishable over $(G,B)$.
\end{lemma}
\begin{proof}
First observe that since $\C T_1$ and $\C T_2$ are determined by their quartets, there exists a subset $A \subset [n]$ with $|A|=4$ such that $\C T_1$ restricted to $A$ and $\C T_2$ restricted to $A$ are distinct four-leaf, binary phylogenetic trees. By Corollary \ref{cor:subnet}, it is sufficient to show that $V_{\C T_1|_A}^{(G,B)} \not\subseteq V_{\C T_2|_A}^{(G,B)}$. Since the dimensions of these varieties are equal (Lemma \ref{lemma:treedim}), this is equivalent to the restricted trees being distinguishable.

We will show that the four leaf binary phylogenetic trees are distinguishable. Let $\C T$ be the four-leaf tree with split $12|34$, and the corresponding interior edge denoted $e_5$. Pick $g, h \in G$ such that $h\not\in [g]$ and consider the polynomial $f = q_{\mathbf{g}}q_{\mathbf{h}} - q_{\mathbf{g'}}q_{\mathbf{h'}}$ where $\mathbf{g} = (g,-g,g,-g), \mathbf{h} = (h,-h,h,-h), \mathbf{g'} = (g, -g,h,-h)$, and $\mathbf{h'}=(h,-h,g,-g)$. We have
$$\psi_{\C T} (f) = a_1^ga_1^ha_2^{-g}a_2^{-h}a_3^ga_3^ha_4^{-g}a_4^{-h}a_5^{0}a_5^{0} - a_1^ga_1^ha_2^{-g}a_2^{-h}a_3^ha_3^ga_4^{-h}a_4^{-g}a_5^{0}a_5^{0} = 0,$$
so that $f \in \ker(\psi_{\C T}) = I_{\C T}^{(G,B)}$.
On the other hand, by looking at the parameters corresponding to the interior edge, the reader can check that $f$ does not belong to the ideals corresponding to the trees with splits $13|24$ and $14|23$ respectively.

In a similar manner one can find polynomials belonging only to the ideal of the tree with split $13|24$ and only to the ideal of the tree with the split $14|23$. It follows that the four leaf binary phylogenetic trees are distinguishable.
\end{proof}

\begin{proposition}\label{prop:sunletident1}
Let $\C N_1$ and $\C N_2$ be two distinct $n$-sunlet networks with $n\geq 5$ and distinct leaves adjacent to the reticulation vertex. Then $\C N_1$ and $\C N_2$ are distinguishable over $(G,B)$.
\end{proposition}
\begin{proof}
By Theorem \ref{thm:sunlet} we have that $\dim V_{\C N_1}^{(G,B)} = \dim V_{\C N_2}^{(G,B)}$. Assume, without loss of generality, that for $\C N_1$ the leaf adjacent to the reticulation vertex is leaf $1$. Let $A=\{2,\ldots,n\}$, so that $\C N_1|_A$ is a caterpillar tree on $n-1$ leaves and $\C N_2|_A$ is an $(n-1)$-sunlet network. Then 
$$\dim V^{(G,B)}_{\C N_1|_A} =  l(2n-5)+1 < l(2n-3) + 1 = \dim V^{(G,B)}_{\C N_2|_A},$$
and so $V^{(G,B)}_{\C N_2|_A} \not\subseteq V^{(G,B)}_{\C N_1|_A}$. By Corollary \ref{cor:subnet}, $\C N_1$ and $\C N_2$ are distinguishable.
\end{proof}

\begin{proposition}\label{prop:sunletident2}
Let $\C N_1$ and $\C N_2$ be two distinct $n$-sunlet networks with $n\geq 4$ such that the leaf adjacent to the reticulation vertex is the same for both networks, and the trees obtained from each network by removing the reticulation vertex and adjacent leaf are distinct. Then $\C N_1$ and $\C N_2$ are distinguishable over $(G,B)$.
\end{proposition}
\begin{proof}

Assume that $\C N_1$ and $\C N_2$ both have leaf $1$ adjacent to the reticulation vertex. Let $A = \{2,\ldots, n\}$ so that by assumption $\C N_1|_A$ and $\C N_2|_A$ are distinct caterpillar trees with $n-1$ leaves. Since these are distinguishable (Lemma \ref{lem:treedis}), the result follows from Corollary \ref{cor:subnet}.
\end{proof}

Observe that Propositions \ref{prop:sunletident1} and \ref{prop:sunletident2} are not sufficient to give identifiability for all sunlet networks. For example, take an $n$ sunlet with leaves labelled in ascending order clockwise around the sunlet with 1 at the reticulation. Then obtain a distinct sunlet by swapping leaves 2 and 3. The caterpillar trees obtained from both of these sunlets by restricting to $\{2,\ldots, n\}$ are the same, so neither Proposition \ref{prop:sunletident1} nor Proposition \ref{prop:sunletident2} applies.
 
 More generally we can give the following identifiability result for triangle-free, level-1 phylogenetic networks. The result relies on the existence of a subset $A$ of the leaf set with particular properties.
 
 \begin{proposition}\label{prop:level1ident}
 Let $\C N_1$ and $\C N_2$ be two triangle-free, level-1 phylogenetic networks on $n$ leaves and both with exactly $c$ cycles, and let $G$ be an abelian group with $|G| > 2$. If there exists a subset $A\subset [n]$ such that either
 \begin{enumerate}
 \item  $\C N_1|_{A}$ and $\C N_2|_{A}$ are triangle-free level-1 phylogenetic networks with distinct number of cycles, or 
 \item  $\C N_1|_{A}$ is a tree and $\C N_2|_{A}$ is a triangle-free level-1 phylogenetic network, or 
  \item  $\C N_1|_{A}$ and $\C N_2|_{A}$ are distinct trees,
 \end{enumerate}
 then $\C N_1$ and $\C N_2$ are distinguishable over $(G,B)$.
 \end{proposition}
 \begin{proof}
First observe that $\dim V^{(G,B)}_{\C N_1} = \dim V^{(G,B)}_{\C N_2}$ by Lemma \ref{lem:level1dim}. Let $\C N_1|_{A}$ and $\C N_2|_{A}$ have $m_1$ and $m_2$ edges respectively, and $c_1$ and $c_2$ cycles respectively.

For case 1, assume without loss of generality that $c_1 < c_2$. Then by Remark \ref{rem:edgenum} we have that $m_1 = 2|A| - 1 + 3c_1$ and $m_2 = 2|A|-1 + 3c_2 $. In particular, $m_1 - c_1 < m_2 - c_2$. Then by Theorem \ref{thm:main} we have that 
$$\dim V^{(G,B)}_{\C N_1|_A} = l(m_1 - c_1) + 1 < l(m_2 - c_2) + 1 = \dim V^{(G,B)}_{\C N_2|_A}.$$ 
It follows that $V^{(G,B)}_{\C N_2|_A} \not\subseteq V^{(G,B)}_{\C N_1|_A}$. For case 2 let us assume that $\C N_1|_{A}$ is a tree and $\C N_2|_{A}$ is a triangle-free level-1 phylogenetic network. Then $\dim V^{(G,B)}_{\C N_1|_A} < \dim V^{(G,B)}_{\C N_2|_A}$ so as above $V^{(G,B)}_{\C N_2|_A} \not\subseteq V^{(G,B)}_{\C N_1|_A}$. For case 3 we have that $V^{(G,B)}_{\C N_1|_A} \not\subseteq V^{(G,B)}_{\C N_2|_A}$ and $V^{(G,B)}_{\C N_2|_A} \not\subseteq V^{(G,B)}_{\C N_1|_A}$ by Lemma \ref{lem:treedis}. In all three cases the result now follows by Corollary \ref{cor:subnet}.
 \end{proof}

  \section{Discussion}
 In this paper we have given a dimension formula for all triangle-free, level-1 phylogenetic networks under a group-based model of evolution. Our main tool was the toric fiber product, for which we gave a dimension formula that we hope will be useful beyond this work.
 
Our results confirmed a conjecture of Gross and Long which states that under the JC model of evolution, the dimensions of large cycle networks (that is, level-1 phylogenetic networks with a single cycle of length at least $4$) are equal \cite[Conjecture~5.1]{GL}. In fact, as we have shown, this is true for all group-based models and level-1 phylogenetic networks where the number of cycles is equal. We were also able to give partial identifiability results for sunlet networks and larger level-1 networks that followed immediately from our results on dimension. 

We were unable to give a general dimension result for 3-sunlets. For this case, our upper bound (Proposition \ref{prop:upper}) still holds, but our proof for the lower bound does not work. This is because with the $\lambda$ we have chosen, when $n=3$ we have only $l$ columns in the matrix $A_\lambda$ coming from $\C T_2$, whilst the rest come from $\C T_1$. Thus the maximum rank of $A_\lambda$ is $\dim V_{\C T_1}^{(G,B)} + l = (2n-2)l + 1$, and this is too small. Nonetheless, we believe the result still holds, and we make the following conjecture.

  \begin{conjecture}\label{conj:3}
  If $\C N$ is the $3$-sunlet network and $|G| > 4$ then
  \[ \dim V^G_{\C N} = lm+1. \]
 \end{conjecture}
 
Our conjecture is backed up by calculations of the dimension $V_{\C N}^{(G,B)}$ for small sunlet networks and small groups. The deficiencies (i.e., the number of dimensions less than the expected dimension $l(2n-1) + 1$) are shown in Table \ref{tab:def}.
 
 \begin{table}[h!]
 \begin{center}
 \caption{Values for the deficiency of $\dim V_{\C N}^{(G,B)}$, where $\C N$ is an $n$-sunlet.}
  \label{tab:def}
 \begin{tabular}{c|ccccccccc}
	  $n$ & $\B Z/2\B Z$ & $\B Z/3\B Z$ & JC & K2P & $(\B Z/2\B Z)^2$ & $\B Z/4\B Z$ & $\B Z/5\B Z$ & $\B Z/6\B Z$ & $\B Z/7\B Z$ \\ \hline
	  3 & $\BF 2$ & $\BF 2$ & $\BF 1$ & $\BF 1$ & $1$ & $0$ & $0$ & $0$ & $0$ \\
	  4 & $1$ & $0$ & $0$ & $0$ & $0$ & $0$ & $0$ & $0$ & $0$ \\
	  5 & $0$ & $0$ & $0$ & $0$ & $0$ & $0$ & $0$ & $0$ & $0$ \\
	  6 & $0$ & $0$ & $0$ & $0$ & $0$ & $0$ & $0$ & $0$ & $0$
	 \end{tabular} 
\end{center}
\end{table}
 Bold values in Table \ref{tab:def} indicate that the variety fills the whole space $\mathbb{C}^{(l+1)^{n-1}}$, and this has dimension less than the expected dimension. Note that for the JC and K2P models we have binomial linear invariants, and it is customary to identify these and reduce the dimension of the ambient space. From Table \ref{tab:def}, it appears that we only have two cases where the dimension of $V_{\C N}^{(G,B)}$ is less than expected for unknown reasons. These are when $G=\mathbb{Z}/2\mathbb{Z}$ and $n=4$, and when $G=\mathbb{Z}/2\mathbb{Z}\times\mathbb{Z}/2\mathbb{Z}$ and $n=3$. The latter case has implications for models of DNA sequence evolution, since the group $G=\mathbb{Z}/2\mathbb{Z}\times\mathbb{Z}/2\mathbb{Z}$ is usually identified with the four nucleic acids, and the corresponding general group-based model of evolution is the Kimura 3-parameter model (K3P). The $3$-sunlet network models events such as hybridisation, so a good understanding of this case will be useful for models in molecular phylogenetics.
 
 A full identifiability result, generalising \cite[Theorem~2]{gross2021distinguishing}, remains open. For the DNA group-based models (JC, K2P, and K3P), one of the key results is that the variety corresponding to the 3-sunlet has smaller dimension than expected. This result can be exploited to give identifiability results on level-1 phylogenetic networks with four leaves (e.g. \cite[Corollary~4.8]{GL}), since for a fixed number of leaves a 3-cycle network will have a strictly lower dimension than a 4-cycle network. For general $G$ however, this is not the case, as shown in Table \ref{tab:def}, so an alternative approach will be necessary to show identifiability for general $G$.
 
 As the authors note in \cite{GL}, this dimension deficiency is in contrast to group-based mixture models, where the number of leaves determines the dimension. Here, we have shown that the dimension of a triangle-free level-1 phylogenetic network variety is fully determined by the number of leaves and the number of cycles (see Theorem \ref{thm:main}), and for large enough $G$ we expect this to be true for all level-1 phylogenetic networks.
\section{Acknowledgements}

Elizabeth Gross is supported by the National Science Foundation (NSF) under grant DMS-1945584. Samuel Martin is supported by the Biotechnology and Biological Sciences Research Council (BBSRC), part of UK Research and Innovation, through the Core Capability
Grant BB/CCG1720/1 at the Earlham Institute, and is  grateful for funding from EPSRC (grant number EP/W007134/1) and BBSRC (grant number BB/X005186/1).  
\bibliographystyle{alpha}
\bibliography{net}
 
\end{document}